\documentclass[11pt]{amsart}

\title{Refining the grading of irreducible Lie colour algebra representations}
\author[M.~Ryan]{Mitchell Ryan}

\usepackage[margin=1in]{geometry}
\usepackage[foot]{amsaddr}
\usepackage{amsmath,amssymb,amsthm}
\usepackage{mathtools}
\usepackage[mathscr]{euscript}
\usepackage{physics}
\usepackage[inline]{enumitem}
\usepackage[colorlinks=true, pdfstartview=FitV, linkcolor=blue, citecolor=blue, urlcolor=blue]{hyperref}
\usepackage{cite}
\usepackage{cleveref}
\usepackage{stmaryrd}
\usepackage{tikz}
\usetikzlibrary{arrows.meta,fit,matrix,shapes.geometric}

\DeclareDocumentCommand\cbrak{ l m m }{\braces#1{\llbracket}{\rrbracket}{#2,#3}} 
\DeclareDocumentCommand\inprod{ l m m }{\braces#1{(}{)}{#2,#3}} 
\let\card\abs							
\newcommand{\bij}{B}						
\newcommand{\CC}{\mathbb{C}}					
\newcommand{\cts}[1][]{\operatorname{C}^{#1}} 			
\DeclareMathOperator{\End}{End}					
\newcommand{\FF}{\mathbb{F}}					
\newcommand{\g}{\mathfrak{g}}					
\newcommand{\gl}{\mathfrak{gl}}					
\DeclareMathOperator{\Id}{Id}					
\DeclareMathOperator{\id}{id}					
\DeclareMathOperator{\im}{im}					
\DeclareMathOperator{\Ind}{Ind}					
\newcommand{\irrep}[1]{\mathsf{Irr}^{\mathsf{tw},\mathsf{sh}}_{#1}}	
\newcommand{\iso}{\cong}					
\newcommand{\niso}{\ncong}					
\newcommand{\loopmodule}[3][]{\mathrm{L}^{#1}_{#3}\hspace{-1pt}{(#2)}}	
\newcommand{\Lp}[1]{L^{#1}}					
\newcommand{\Mat}[1][]{\mathrm{M}_{#1}}					
\newcommand{\orthog}[1]{#1^{\perp}}				
\newcommand{\osp}{\mathfrak{osp}}				
\let\pdual\widehat 						
\newcommand{\RR}{\mathbb{R}}					
\newcommand{\R}{\mathscr{R}}					
\newcommand{\restrict}[2]{#1|_{#2}}
\newcommand{\sbrak}[2]{\cbrak{#1}{#2}^{\sigma}} 		
\newcommand{\spl}[1][]{\mathfrak{sl}_{#1}}			
\newcommand{\splc}[1][]{\spl[#1]^\mathrm{c}}			
\newcommand{\so}[1][]{\mathfrak{so}_{#1}}			
\newcommand{\sotwoc}[1][]{\mathfrak{so}^c_{\oplus#1}}			
\newcommand{\sotwo}[1][]{\mathfrak{so}_{\oplus#1}}			
\DeclareMathOperator{\spn}{span}				
\newcommand{\trsp}[1]{#1^\mathsf{T}}				
\newcommand{\Uea}{U}						
\newcommand{\units}[1]{#1^{\times}}				
\newcommand{\ZZ}{\mathbb{Z}}					
\newcommand{\Ztwo}[1][]{\ZZ_2^{#1}}				
\newcommand{\Ztzt}{\Ztwo\times\Ztwo}				
\newcommand{\Hilb}{\mathscr{H}}					

\newcommand{\pdualtwist}{\pdual{\Gamma}/\orthog{H}}		

\newcommand{\dfn}[1]{{\color{red!70!black}\itshape #1}}

\theoremstyle{plain}
\newtheorem{thm}{Theorem}[section]
\newtheorem{lem}[thm]{Lemma}
\newtheorem{prop}[thm]{Proposition}

\theoremstyle{definition}
\newtheorem{defn}[thm]{Definition}

\theoremstyle{remark}
\newtheorem{rmk}[thm]{Remark}

\usepackage[colorinlistoftodos]{todonotes}
\definecolor{lightblue}{RGB}{137,207,240}

\AddToHook{env/lem/begin}{\crefalias{thm}{lem}}
\AddToHook{env/prop/begin}{\crefalias{thm}{prop}}
\AddToHook{env/cor/begin}{\crefalias{thm}{cor}}

\AddToHook{env/defn/begin}{\crefalias{thm}{defn}}
\AddToHook{env/ex/begin}{\crefalias{thm}{ex}}
\AddToHook{env/rmk/begin}{\crefalias{thm}{rmk}}

\crefname{ex}{example}{examples}
\Crefname{ex}{Example}{Examples}
\crefname{thm}{theorem}{theorems}
\Crefname{thm}{Theorem}{Theorems}
\crefname{prop}{proposition}{propositions}
\Crefname{prop}{Proposition}{Propositions}
\crefname{lem}{lemma}{lemmas}
\Crefname{lem}{Lemma}{Lemmas}

\begin{document}
\address{Department of Mathematics, University of British Columbia, Vancouver, BC, Canada\\
ORCID: 0009-0006-2038-4410}
\email{\href{mailto:mitchell.ryan@math.ubc.ca}{mitchell.ryan@math.ubc.ca}}

\keywords{Color Lie (super)algebras, Graded Lie (super)algebras, Irreducible representations}
\subjclass[2020]{17B75, 17B70, 17B10}

\begin{abstract}
	We apply the loop module construction of Mazorchuk--Zhao in the context of Lie colour algebras.
	We construct a bijection between the equivalence classes of all finite-dimensional graded irreducible Lie colour algebra representations from the irreducible representations for Lie superalgebras.
	This bijection is obtained by applying the loop module construction iteratively to cyclic groups of prime order in the Jordan--H\"older decomposition of the grading group.
	Restricting to cyclic groups of prime order in this way greatly simplifies the construction.
	Despite the bijection between Lie colour algebra representations and Lie superalgebra representations, Lie colour algebras maintain a non-trivial representation theory distinct from that of Lie superalgebras.
	We demonstrate the applicability of the loop module construction to Lie colour algebras in three examples:
	a Hilbert space for a quantum mechanical model, representations of a colour version of \( \mathfrak{sl}_2 \) and representations of a class of \( \mathbb{Z}_2\times\mathbb{Z}_2 \)-graded colour algebras of type D.
\end{abstract}

\maketitle
\section{Introduction}
Lie colour algebras are a generalisation of Lie superalgebras to grading by any abelian group \( \Gamma \) (rather than just \( \Ztwo \))
and were introduced by Rittenberg and Wyler~\cite{RW1978a,RW1978b} in 1978 (however, see also~\cite{Ree1960}).
Lie colour algebras initially found applications in (for example): de Sitter spaces~\cite{LR1978,Vasiliev1985,Tolstoy2014a}, quasispin~\cite{JYW1987}, strings~\cite{Zheltukhin1987} and extensions of Poincar\'e algebras~\cite{WillsToro1995,WillsToro2001}.
There has been a recent increase in activity surrounding applications of \( \Ztzt \)-graded Lie colour algebras.
This recent activity has mainly been focussed on parastatistics~\cite{YJ2001,JYL2001,KHA2011a,KHA2011b,Tolstoy2014b,SVdJ2018,Toppan2021a,Toppan2021b,Zhang2023,SVdJ2024} and \( \Ztzt \)-graded quantum mechanics~\cite{AKTT2016,AKTT2017,BD2020a,AKT2020,AKT2021,AAD2020a,AAD2020b,DA2021,AAD2021,Bruce2021,Quesne2021,AIKT2023}.

Many Lie colour algebra applications make use of irreducible representations to construct various models.
A classification of such representations would therefore be useful.
Indeed, in~\cite{AD2022}, the authors conclude their paper with a desire for a classification of all irreducible representations for the \( \Ztzt \)-supersymmetry algebra with \( \mathcal{N}\geq 2 \).
Such a classification motivates the results of this paper.

Various classifications of the algebras themselves under different assumptions have already been performed~\cite{Silvestrov1997,KT2021,BP2009,LuTan2023,SVdJ2023}.
Moreover, the representation theories of the colour versions of the following algebras have been examined: \( \gl \), \( \spl \), \( \osp \)~\cite{GJ1983,AA2021}, the Lie algebra for the group of plane motions~\cite{Silvestrov1996}, the Heisenberg Lie algebra~\cite{SS2006}, super Schr\"odinger algebras~\cite{AS2017} and a supersymmetry algebra~\cite{AD2022}.

One of the most powerful techniques in the study of Lie colour algebras is that of discolouration,
which provides a bijection between Lie colour algebras and \( \Gamma \)-graded Lie superalgebras.
This bijection preserves much of the important structure, such as subalgebras, ideals and subrepresentations.
Discolouration, in the specific case of \( \Ztwo[n] \)-gradings, was introduced concurrently with the algebras themselves~\cite{RW1978a},
was soon generalised~\cite{Scheunert1979} to grading by any finitely generated abelian group,
and further generalised to an arbitrary abelian group independently in \cite{Passman1998,BahturinMontgomery1999}.

Discolouration was reformulated in~\cite{Kleeman1985}, using so-called Klein operators lying outside the universal enveloping algebra.
It has since been shown that these Klein operators can be expressed in an algebraic extension of the universal enveloping algebras for a colour version of \( \gl \)~\cite{MB1997} or for the non-coloured supersymmetric \( \gl(m|n) \)~\cite{ISVdJ2020} and that Klein operators appear in the \( q\to -1 \) limit of \( U_q(\spl[3]) \)~\cite{AI2024}.
This version of discolouration has already been applied to the study of the representation theory.
For example, in~\cite{ISVdJ2020} it is shown how the covariant representations of \( \gl(m|n) \) can be lifted to the colour version of \( \gl \).

Despite widespread awareness of the results of~\cite{Scheunert1979},
discolouration has rarely been utilised in the recent literature surrounding applications of \( \Ztzt \)-graded Lie colour algebras.
Perhaps a partial explanation lies with the following nuance:
the discolouration of a Lie colour algebra representation is a \( \Gamma \)-graded representation,
whose representation theory can be quite different from the more familiar \( \Ztwo \)-graded representation theory of Lie superalgebras.
For example, a \( \Gamma \)-graded irreducible representation may become reducible when considering only the \( \Ztwo \)-grade.

Concurrently (yet largely independently) gradings on (non-colour) Lie algebras have been systematically studied,
beginning with~\cite{PZ1989}.
The classification of fine gradings on finite-dimensional simple Lie algebras over an algebraically closed field with characteristic zero has been completed---see the monograph~\cite{EK2013} and references therein, together with~\cite{DV2016,Yu2016,Elduque2016}.
Gradings over fields with nonzero characteristic (see e.g.~\cite{EK2013}) and over non-algebraically closed fields (see e.g.~\cite{EKRE2022}) are also under study.

Given that some Lie colour algebras discolour to give graded Lie algebras (i.e.\ graded Lie superalgebras with trivial odd sector),
the classification results of~\cite{EK2013,DV2016,Yu2016,Elduque2016} for gradings on Lie algebras are also useful for classifying Lie colour algebras.
Similarly, a classification of gradings on Lie superalgebras would aid in classifying Lie colour algebras (see~\cite{HDSK2019} and references therein for work in the direction of classifying Lie superalgebra gradings).
Application of the grading classification results to the problem of classifying Lie colour algebras has been studied in~\cite{BP2009}.

Although the research into graded algebras has mainly focussed on classifying the fine gradings,
the representation theory of graded algebras has also received attention.
Various works have shown how one can construct the graded irreducible modules from ungraded modules via different methods:
induced modules~\cite{EK2015},
thin coverings~\cite{BL2007} or loop modules~\cite{MZ2018}.

In this paper, we will focus on the loop module construction,
since loop modules are a natural structure to consider in the study of Lie colour algebras.
Indeed, loop constructions have already appeared in the Lie colour algebra literature,
often under a different name or left unnamed.
For example, the `enhancement' of a Lie colour algebra defined in~\cite{Price2024} is a twisted version of the loop algebra studied in~\cite{ABFP2008}.
Additionally, we will show (\Cref{sec:loopquantum}) that the \( \Ztzt \)-graded quantum mechanical Hilbert space realisation in~\cite{BD2020a} can be constructed as a loop module. 

The initial motivation for introducing loop modules was to study graded irreducible representations:
every \( \Gamma \)-graded irreducible representation is the loop module of some ungraded irreducible representation~\cite{MZ2018}.
Moreover, the loop module can be used to construct an equivalence between two categories (see~\cite[Theorem~4.16]{EK2017} for details)
and a bijection between equivalence classes of \( \Gamma \)-graded irreducible and \( \Gamma/H \)-graded irreducible modules~\cite[Remark~7.2]{EK2017}.
Although these results were initially proved in the context graded (non-colour) Lie algebras,
the results still hold for Lie colour algebras.
We make extensive use of these results throughout this paper.

The goal of this paper is twofold:
\begin{enumerate*}
	\item\label{item:applyloop}to apply the loop module construction to Lie colour algebras and demonstrate its utility in this context; and%
	\item\label{item:constructiveloop}to present a more constructive loop module procedure to simplify computations.
\end{enumerate*}

For goal~(\ref{item:applyloop}),
we have already mentioned that the loop module already appears in the Lie colour algebra literature through enhancements~\cite{Price2024} and a \( \Ztzt \)-graded quantum mechanical Hilbert space~\cite{BD2020a}.
In addition, we apply the loop module construction (from goal~(\ref{item:constructiveloop})) to a colour version of \( \spl[2] \) to derive all of its irreducible modules.
We compare these irreducible colour \( \spl[2] \) modules with those that have already appeared in the literature~\cite{CSVO2006,BS2022}.

The key idea for goal~\ref{item:constructiveloop}
is to take iterated loop modules, only increasing the size of the grading group by a cyclic group of prime order \( H \) at each step.
This simplifies the construction due to the following result:
for a finite-dimensional \( \Gamma/H \)-graded irreducible module \( V \), either \( V \) or its loop module (but not both) is a \( \Gamma \)-graded irreducible module.
If the desired grading group is finite (which is often the case in applications of Lie colour algebras)
we can use the Jordan--H\"older decomposition of the grading group to build up to the full grading group.

This iterated loop module process provides a method for deriving every finite-dimensional \( \Gamma \)-graded irreducible representation from the corresponding \( \Ztwo \)-graded representations.
In particular, we obtain bijections \( \bij_1 \) and \( \bij_2 \) that map between equivalence classes (with equivalence relations weaker than isomorphism) of irreducible representations, as per the following diagram.
\begin{center}
	\begin{tikzpicture}[xscale = 4, yscale = 1.5]
		\node (Z2super) at (-1,1) {\( \Ztwo \)-graded superalgebra irreps};
		\node (Gsuper) at (1,1) {\( \Gamma \)-graded superalgebra irreps};
		\node (Gcolour) at (1,-1) {\( \Gamma \)-graded colour algebra irreps};
		\node (uncolour) at (-1,-1) {ungraded colour algebra irreps};

		\draw [{<[scale=1.5]}-{>[scale=1.5]}] (Z2super) -- node[above]{\( \bij_1 \)} (Gsuper);
		\draw [-{>[scale=1.5]}] (Gsuper.300) -- node[right]{recolouration} (Gcolour.60);
		\draw [{<[scale=1.5]}-] (Gsuper.240) -- node[left]{discolouration} (Gcolour.120);
		\draw [{<[scale=1.5]}-{>[scale=1.5]}] (uncolour) -- node[above]{\( \bij_2 \)} (Gcolour);
	\end{tikzpicture}
\end{center}
These bijections are provided in~\Cref{thm:bijection}
(the existence of similar bijections were noted in~\cite[Remark~7.2]{EK2017}).
The \( \Ztwo \)-graded irreducible superalgebra representations have been well-studied,
and these bijections will allow us to carry across these results to Lie colour algebras.
It is important to note that these bijections are more complicated than discolouration and allow for interesting differences in the representation theories.

This paper is organised as follows.
In \Cref{sec:prelim}, we review some known results for Lie colour algebras, and results analogous to those for Lie (super)algebras.
We then summarise the loop module construction in \Cref{sec:loopmodule}.
As an example, \Cref{sec:loopquantum} uses the loop module to construct the Hilbert space of the \( \Ztzt \)-graded supersymmetric system of~\cite{BD2020a}.
In \Cref{sec:loopirrep}, we show how an iterated loop module construction can be used to derive the finite-dimensional \( \Gamma \)-graded irreducible modules from the \( \Gamma/H \)-graded modules.
We give an example of deriving the finite-dimensional irreducible modules of a colour version of \( \spl[2] \) (\Cref{sec:exampleslc2}) and compare with~\cite{CSVO2006,BS2022}.
In addition, we give a more involved example in \Cref{sec:egD} where we classify the graded and ungraded irreducible modules of a class of type D \( \Ztzt \)-graded Lie colour algebras.
We end with some concluding remarks in \Cref{sec:conclusion}.

\section{Lie colour algebras and their representations}\label{sec:prelim}
A Lie colour algebra generalises a Lie superalgebra to grading by some abelian group \( \Gamma \) (instead of just \( \Ztwo \)).
Note that we will use additive notation for abelian groups throughout this paper.
For the formal definition of a Lie colour algebra, we refer the reader to~\cite{Scheunert1979}.
As an illustrative example:
if \( \R = \bigoplus_{\alpha\in\Gamma}\R_{\alpha} \) is an associative algebra such that \( \R_{\alpha}\R_{\beta}\subseteq \R_{\alpha+\beta} \) then \( \R \) can be given the structure of a Lie colour algebra with bracket
\[
	\cbrak{x}{y} = xy-\varepsilon(\alpha,\beta)yx, \qquad x\in \R_{\alpha},\,y\in \R_{\beta}.
\]
Here, \( \varepsilon\colon\Gamma\times\Gamma\to\FF \) (where \( \FF \) is the field of scalars) is called a commutation factor~\cite{Scheunert1979} (or antisymmetric bicharacter~\cite{CSVO2006}, or phase function~\cite{MB1997}).
That this generalizes a Lie superalgebra is apparent:
a Lie superalgebra is simply a Lie colour algebra graded by the group \( \Gamma=\Ztwo \) with commutation factor \( \varepsilon(\alpha,\beta)=(-1)^{\alpha\cdot\beta} \).
Note that, in general, a commutation factor must satisfy certain conditions (see~\cite{Scheunert1979}).
Every Lie colour algebra is contained in a \( \Gamma \)-graded associative algebra:
we can construct the universal enveloping algebra \( \Uea(\g) \) in the usual way
and prove that \( \g \) embeds inside \( \Uea(\g) \) using a generalised version of the Poincar\'e--Birkhoff--Witt Theorem; see~\cite[Section~4]{Scheunert1979} for more details.

If \( \varepsilon(\alpha,\beta) = 1 \) or \( \varepsilon(\alpha,\beta) = -1 \)
then we will often write \( \cbrak{x}{y} = \comm{x}{y} \) or \( \cbrak{x}{y} = \acomm{x}{y} \) (respectively).
This is analogous to the notation used for Lie superalgebras.

We can define notions such as homomorphisms, subalgebras and representations similarly to superalgebras.
However, to retain information about the \( \Gamma \)-grading, we require some maps
(such as homomorphisms, intertwiners, and inclusion maps of subalgebras and \( \Gamma \)-graded submodules)
to be homogeneous of degree \( 0 \), as per the following definition.
\begin{defn}
	Let \( V=\bigoplus_{\gamma\in\Gamma} V_{\gamma} \) and \( W=\bigoplus_{\gamma\in\Gamma}W_{\gamma} \) be vector spaces over \( \FF \).
	A linear map \( f\colon V\to W \) is \dfn{homogeneous of degree \( \gamma\in\Gamma \)} if
	\( f(V_{\alpha}) \subseteq W_{\alpha+\gamma} \) for all \( \alpha\in\Gamma \).
	(In particular, \( f \) is a homogeneous element of the Lie colour algebra \( \gl(V,\varepsilon) \), see~\cite{Scheunert1979}.)
\end{defn}


\subsection{Representations}
Although representations can be defined similarly to those for superalgebras, we make an important distinction between graded and ungraded representations.

\begin{defn}
	An \dfn{ungraded representation} of a Lie colour algebra \( \g \) on a vector space \( V \) is a linear map \( \rho\colon\g\to\End(V) \)
	satisfying \( \rho(\cbrak{x}{y}) = \rho(x)\rho(y) - \varepsilon(\alpha,\beta)\rho(y)\rho(x) \) for all \( x\in \g_\alpha \), \( y\in\g_\beta \).
	A \dfn{\( \Gamma \)-graded representation} of \( \g \) on a \( \Gamma \)-graded vector space \( V = \bigoplus_{\gamma\in\Gamma} V_\gamma \) is an ungraded representation that also satisfies
	\( \rho(\g_\alpha)V_\gamma \subseteq V_{\alpha+\gamma} \).
	We call \( V_\gamma \) the \dfn{\( \gamma \)-sector}
	and the elements of \( V_\gamma \) are said to have \dfn{degree \( \gamma \)}.
	In both the ungraded and graded cases, we say that \( V \) is a \dfn{\( \g \)-module given by \( \rho \)} (or simply a \dfn{\( \g \)-module}).
	We will often omit the representation map \( \rho \).
\end{defn}

In the majority of this paper,
we will instead work with representations of
\( \Gamma \)-graded associative algebras,
i.e.\ representations \( \rho\colon \R \to \End(V) \) where
\( \R = \bigoplus_{\alpha\in\Gamma}\R_\alpha \) is an associative algebra
that satisfies \( \R_\alpha \R_\beta \subseteq \R_{\alpha+\beta}  \).
Such a representation is \dfn{graded}
if \( V = \bigoplus_{\gamma\in \Gamma} V_\gamma \)
and \( \rho(\R_\alpha) V_\gamma \subseteq V_{\alpha+\gamma} \).
In fact, Lie colour algebra representations
are a special case of representations for a graded associative algebra:
the category of ungraded (resp.\ graded) representations of a Lie colour algebra \( \g \) is isomorphic to the category of ungraded (resp.\ graded) representations of 
its universal enveloping algebra \( \Uea(\g) \).
This follows immediately from the universal property of \( \Uea(\g) \)
\cite[Proposition~1]{Scheunert1979} and the fact that \( \g\subset \Uea(\g) \) generates \( \Uea(\g) \).

We do not require maps related to ungraded representations (such as intertwiners and submodules) to be homogeneous of degree \( 0 \),
but we do require this condition of the corresponding maps for graded representations.

For \( \R \) a \( \Gamma \)-graded associative ring (such as \( \Uea(\g) \)),
we call a (\( \Gamma \)-graded) \( \R \)-module \dfn{(\( \Gamma \)-graded) irreducible} if it has no proper non-trivial (\( \Gamma \)-graded) submodules.
Note that there is an important distinction between the definitions of (ungraded) irreducible and \( \Gamma \)-graded irreducible: a module may be \( \Gamma \)-graded irreducible but not ungraded irreducible!

Recall that an \( \R \)-module is called completely reducible if it can be written as a direct sum of irreducible submodules.
Weyl's Theorem on complete reducibility does not hold for semisimple Lie colour algebras (indeed, Weyl's Theorem does not hold for semisimple Lie superalgebras~\cite{Kac1977}).
As such, we need to manually determine when an \( \R \)-module is completely reducible.
The following propositions will be of great help.

\begin{prop}\label{prop:sumsimpledirect}
	Consider a (\( \Gamma \)-graded or ungraded) \( \R \)-module \( V \) that can be written as \( V=\sum_{i\in I} V_i \) for irreducible \( \R \)-modules \( V_i \) (\( i\in I \)).
	Then, there exists some subset \( J \subseteq I \) such that \( V = \bigoplus_{i\in J} V_i \).
\end{prop}

\begin{prop}\label{prop:semisimplesubset}
	Submodules and quotient modules of a (\( \Gamma \)-graded or ungraded) completely reducible \( \R \)-module are completely reducible.
\end{prop}

These propositions are well-known facts from the module theory of rings (see e.g.~\cite[Chapter~XVII, Lemma~2.1 and Proposition~2.2]{Lang2002}).

\subsection{Graded quotient modules}
Let \( V=\bigoplus_{\gamma\in\Gamma}V_{\gamma} \) be a \( \Gamma \)-graded \( \R \)-module given by the representation \( \rho \) and \( U=\bigoplus_{\gamma\in\Gamma}U_{\gamma} \) a \( \Gamma \)-graded submodule of \( V \) (so that \( U_{\gamma}\subseteq V_{\gamma} \)).

\begin{defn}
	The \dfn{graded quotient module of \( V \) by \( U \)},
	denoted \( V/U \), is a direct sum of quotient vector spaces \( \bigoplus_{\gamma\in\Gamma}V_{\gamma}/U_{\gamma} \) with \( \R \)-module structure given by \( \pi\colon \R\to\End(V/U) \)
	defined by
	\[
		\pi(x)(v+U_{\gamma}) = \rho(x)v+U_{\alpha+\gamma}
		\qquad \text{for}~x\in \R_{\alpha},\, v\in V_{\gamma}.
	\]
\end{defn}

The proof that the above representation is well defined is similar to that for superalgebras.
The only additional property to check is whether \( \pi \) is homogeneous of degree \( 0 \), which is easily verified.

The three isomorphism theorems hold for \( \R \)-modules.
Again, the proofs are similar to superalgebras, needing only to verify that all maps are homogeneous of degree \( 0 \).

\subsection{Discolouration}

Scheunert~\cite{Scheunert1979} gave a bijection between Lie colour algebras and graded Lie superalgebras and between their representations.
This bijection has many nice properties; for example, it preserves subalgebras, ideals and subrepresentations.
To compute the image of a colour algebra under this bijection, we deform the bracket using a multiplier.

Given a Lie colour algebra \( \g \)
and a multiplier \( \sigma\colon\Gamma\times\Gamma\to\units{\FF} \)
(for the definition, see~\cite{Scheunert1979}; also called a multiplicative \( 2 \)-cocycle~\cite{BP2009}),
we can define a new bracket \( \sbrak{\cdot}{\cdot} \) by
\[
	\sbrak{x}{y} = \sigma(\alpha,\beta)\cbrak{x}{y}, \qquad x\in\g_{\alpha},y\in\g_{\beta}.
\]
We can choose this multiplier \( \sigma \) so that \( \g \) equipped with \( \sbrak{\cdot}{\cdot} \) forms a Lie superalgebra \cite{Scheunert1979,Passman1998,BahturinMontgomery1999}.
We denote this Lie superalgebra \( \g^\sigma \).
The map \( \g\mapsto\g^{\sigma} \)
is a bijection between the Lie colour algebras with commutation factor \( \varepsilon \)
and the \( \Gamma \)-graded Lie superalgebras~\cite{Scheunert1979}.
As such, we call \( \g^{\sigma} \) a \dfn{discolouration} of \( \g \).
Discolouration is an invertible procedure, and we will call \( (\g^{\sigma})^{1/\sigma} = \g \) the \dfn{recolouration} of \( \g^{\sigma} \).


Note that there is an alternative method of discolouration which works by adding new elements, called Klein operators, to the universal enveloping algebra (see~\cite{Kleeman1985,MB1997, ISVdJ2020, AI2024}).
However, for the remainder of this paper we will use the original discolouration process of~\cite{Scheunert1979} described in this section.

It is also possible to discolour representations.
Given a graded representation \( \rho\colon\g\to\gl(V,\varepsilon) \)
and a discolouring multiplier \( \sigma\colon\Gamma\times\Gamma\to\units{\FF} \),
we can define the discolouration
 of \( \rho \) to be the graded representation \( \rho^{\sigma} \) of \( \g^{\sigma} \) given by
\[
	\rho^{\sigma}(x)v = \sigma(\alpha,\gamma)\rho(x)v, \qquad \text{for}~x\in\g_{\alpha},v\in V_{\gamma}.
\]
Similarly, the map \( \rho\mapsto\rho^\sigma \) is a bijection between
the graded representations of \( \g \) 
and the graded representations of the corresponding discoloured \( \Gamma \)-graded Lie superalgebra~\cite{Scheunert1979}.
Thus, the representation theory of Lie colour algebras
can be completely derived from that of \( \Gamma \)-graded Lie superalgebras.

The \( \Ztwo \)-graded irreducible representations of Lie superalgebras have been well studied.
However, these results cannot be immediately applied to find \( \Gamma \)-graded irreducible representations.
For instance, it might not be possible to find a \( \Gamma \)-grading for a \( \Ztwo \)-graded representation.
On the other hand,
there might be a \( \Gamma \)-graded irreducible representation
that becomes reducible when considered as a \( \Ztwo \)-graded representation.

\section{loop module}\label{sec:loopmodule}
Throughout the remainder of this paper,
we will assume that \( \Gamma \) is a finite abelian group
and that the ambient field \( \FF \) is algebraically closed 
of characteristic \( 0 \).

Let \( \R = \bigoplus_{\gamma\in\Gamma}\R_{\gamma} \) be a \( \Gamma \)-graded associative algebra and let \( H \) be a subgroup of the abelian group \( \Gamma \).
Notice that \( \R \) has a natural \( \Gamma/H \) grading, given by
\( \R=\bigoplus_{\Lambda\in\Gamma/H}\R^{\Gamma/H}_{\Lambda} \) where \( \R^{\Gamma/H}_{\Lambda}=\bigoplus_{\alpha\in\Lambda}\R_{\alpha} \).

\begin{defn}[Mazorchuk--Zhao~\cite{MZ2018}; see also~\cite{ABFP2008}]\label{dfn:loopmodule}
	Let \( V=\bigoplus_{\Lambda\in\Gamma/H}V_{\Lambda} \) be a \( \Gamma/H \)-graded \( \R \)-module
	given by \( \rho \).
	Recall that the group algebra \( \FF\Gamma \) is the algebra with basis \( \{e_{\gamma}\mid \gamma\in\Gamma\} \) 
	and algebra product defined by \( e_{\alpha}\cdot e_{\beta} = e_{\alpha+\beta} \).
	The vector space \( V\otimes\FF\Gamma \) can be given an \( \R \)-module structure by
	\[
		x(v\otimes e_{\beta}) = xv\otimes e_{\alpha+\beta} 
	\]
	for~\(x\in\R_{\alpha}\) and \( v\otimes e_{\beta}\in V\otimes \FF\Gamma \).
	The \dfn{Loop module of \( V \) by \( H \)} is the \( \Gamma \)-graded \( \R \)-module
	\[
		\loopmodule[\Gamma]{V}{H} = \bigoplus_{\gamma\in\Gamma} V_{\gamma + H} \otimes e_{\gamma} \subset V\otimes \FF\Gamma.
	\]
	The \( \gamma \)-sector of \( \loopmodule[\Gamma]{V}{H} \) is \( \loopmodule[\Gamma]{V}{H}_{\gamma} = V_{\gamma + H}\otimes e_\gamma \).
\end{defn}

Not every \( \Gamma/H \)-graded module \( V \) can be given a \( \Gamma \)-grading.
Intuitively,
the loop module adds in multiple copies of \( V \) to ensure that there is enough `space' for a refined \( \Gamma \)-grading.
If \( V \) can be given a \( \Gamma \)-grading, then the loop module contains \( V \):


\begin{prop}\label{prop:gradesubrep}
	If \( V \) can be given a \( \Gamma \)-grading that refines the \( \Gamma/H \)-grading,
	then \( V \) is a \( \Gamma \)-graded submodule of \( \loopmodule[\Gamma]{V}{H} \).
\end{prop}
\begin{proof}
	Let \( V = \bigoplus_{\gamma\in\Gamma}W_{\gamma} \) be the \( \Gamma \)-grading,
	with \( V_{\Lambda} = \bigoplus_{\gamma\in\Lambda}W_{\gamma} \) for all \( \Lambda\in\Gamma/H \)
	(i.e.\ the \( \Gamma \)-grading refines the \( \Gamma/H \) grading).
	Define a map \( \Phi\colon V\to \loopmodule[\Gamma]{V}{H} \) by \(\Phi(w) = w\otimes e_{\gamma} \) for \( w\in W_{\gamma} \)
	and extending to inhomogeneous elements by linearity.
	Linearity and injectivity of \( \Phi \) follow immediately.
	The map \( \Phi \) is homogeneous of degree \( 0 \) by definition.
	And \( \Phi \) is an intertwiner:
	for \( w\in W_{\gamma} \) and \( x\in\R_{\alpha} \),
	\[
		\Phi(xw) 
		= xw\otimes e_{\alpha+\gamma}
		= x(w\otimes e_{\gamma})
		= x(\Phi(w))
		\qedhere
	\]
\end{proof}

As well as a \( \Gamma \)-grading,
\( \loopmodule[\Gamma]{V}{H} \) has a natural \( \Gamma/H \)-grading
given by
\[
	\loopmodule[\Gamma]{V}{H} = \bigoplus_{\Lambda\in\Gamma/H} \left(\bigoplus_{\gamma\in\Lambda} V_{\Lambda}\otimes e_{\gamma}\right).
\]

\begin{rmk}\label{rmk:enhancement}
	If \( \Gamma = K \times H \) for some group \( K \),
	then we could equivalently define the loop module 
	to be \( V \otimes \FF[H] \)
	with representation given by \( x(v\otimes e_h) = xv\otimes (e_{h + \eta}) \)
	for \( x\in\R_{(\kappa,\eta)} \) and \( v\otimes e_h\in V\otimes\FF[H] \).
	Compare with~\cite{Price2024}, where this construction, together with a twist by a multiplier \( \sigma \) on \( \FF[H] \), is called an \dfn{enhancement}.
	The enhancement approach is used in~\cite{BP2009} and~\cite{Price2024} to study the structure of simple Lie colour algebras themselves (rather than their modules).
	An enhancement of a Lie colour algebra is similar to the loop algebra construction studied elsewhere~\cite{ABFP2008,MZ2018}.
\end{rmk}

\subsection{loop module and quantum mechanical systems}\label{sec:loopquantum}
A common way to obtain a \( \Gamma \)-graded colour supersymmetric quantum mechanical system is by modifying an existing supersymmetric system.
This is the approach taken in (for example) \cite{BD2020a}, where 
the \( \Ztzt \)-graded supersymmetric system
presented in~\cite[Section~III]{BD2020a}
can be obtained by introducing an additional central charge to Witten's model~\cite[Section~6]{Witten1981}.
In~\cite{BD2020a}, the Hilbert space realisation for this \( \Ztzt \)-graded system was constructed independently from that of Witten's model.
In this section, we will demonstrate that
this \( \Ztzt \)-graded Hilbert space realisation can be constructed from Witten's realisation using a loop module.
In this way, the loop module has already appeared in the Lie colour algebra literature and is a natural structure to consider in the study of Lie colour algebras.

%
%

Witten considered~\cite[Section~6]{Witten1981} a superalgebra with even sector spanned by a Hamiltonian \( H \) and odd sector spanned by two supercharges \( Q_1 \) and \( Q_2 \) satisfying
\[
	\comm{Q_i}{H} = 0 \qquad \acomm{Q_i}{Q_j} = \delta_{ij}H
\]
for \( i,j=1,2 \).
This algebra was realised as operators acting on two component Pauli spinors in the Hilbert space \( \Hilb^\mathrm{W} = \Lp{2}(\RR)\otimes\CC^2 \) via
\begin{align*}
	H &= \frac{1}{2}\left(p^2+W^2(x)+\hbar\sigma_3\dv{W}{x}\right), &
	Q_1 &= \frac{1}{2}(\sigma_1p+\sigma_2W(x)), &
	Q_2 &= \frac{1}{2}(\sigma_2p-\sigma_1W(x))
\end{align*}
where \( \sigma_i \) (\( i=1,2,3 \)) are the Pauli matrices, \( W \in\cts[\infty](\RR) \) and \( p=-i\hbar\dv*{x} \) is the momentum operator.
This Hilbert space representation is a \( \Ztwo \)-graded representation of the superalgebra, with grading given by \( \CC^2=\CC_0\oplus\CC_1 \)
(that is, the first component of the spinor has an even grade and the second component has an odd grade).

By adding a central charge,
\[
	Z = -\frac{i}{2}\left(\sigma_3(p^2+W^2(x))+\hbar\dv{W}{x}\right)
\]
the original superalgebra becomes a \( \Ztzt \)-graded Lie colour algebra (without requiring recolouration) with commutation factor \( \varepsilon(\alpha_1\alpha_2,\beta_1\beta_2) = (-1)^{\alpha_1\beta_1 + \alpha_2\beta_2} \) for \( \alpha_1\alpha_2,\,\beta_1\beta_2\in \Ztzt = \{00,01,10,11\} \)  and with sectors spanned by
\[
	00\text{-sector: } H \qquad 01\text{-sector: } Q_1 \qquad 11\text{-sector: } Z \qquad 10\text{-sector: } Q_2 
\]
and satisfying the following additional relations
\[
	\comm{Q_2}{Q_1} = Z \qquad \acomm{Q_i}{Q_i} = H
\]
while all other Lie colour brackets vanish.
We can easily verify that \( \Hilb^\mathrm{W} \) becomes an ungraded module of the Lie colour algebra.

Taking the loop module \( \loopmodule[\Ztzt]{\Hilb^\mathrm{W}}{\Ztwo} \) turns \( \Hilb^\mathrm{W} \) into a \( \Ztzt \)-graded module.
Here, we use \( \Ztwo\iso\{00,11\}\leq\Ztzt \).
Set
\[
	\Hilb_{00} = \Hilb^\mathrm{W}_0\otimes e_{00} \qquad
	\Hilb_{11} = \Hilb^\mathrm{W}_0\otimes e_{11} \qquad
	\Hilb_{01} = \Hilb^\mathrm{W}_1\otimes e_{01} \qquad
	\Hilb_{10} = \Hilb^\mathrm{W}_1\otimes e_{10}.
\]
Then, the loop module is \( \loopmodule[\Ztzt]{\Hilb^\mathrm{W}}{\Ztwo}=\Hilb_{00}\oplus\Hilb_{01}\oplus\Hilb_{11}\oplus\Hilb_{10}=\Lp{2}(\RR)\otimes\CC^4 \) (note the order of the sectors).
We can then explicitly compute the representation on four-component spinors:
\begin{align*}
	\rho^{\Ztwo}(H)&=
	\begin{pmatrix}
		H & 0\\
		0 & H
	\end{pmatrix},
	&
	\rho^{\Ztwo}(Q_1)&=
	\begin{pmatrix}
		Q_1 & 0\\
		0 & Q_1
	\end{pmatrix},
	&
	\rho^{\Ztwo}(Q_2)&=
	\begin{pmatrix}
		0 & Q_2\\
		Q_2 & 0
	\end{pmatrix},
	&
	\rho^{\Ztwo}(Z)&=
	\begin{pmatrix}
		0 & Z\\
		Z & 0
	\end{pmatrix}
\end{align*}
where \( \rho^{\Ztwo} \) is the representation for \( \loopmodule[\Ztzt]{\Hilb^\mathrm{W}}{\Ztwo} \).
Setting \( H_{00}=2\rho^{\Ztwo}(H),\, Q_{01}=\rho^{\Ztwo}(Q_1),\, Q_{10}=\rho^{\Ztwo}(Q_2) \) and \( Z_{11}=2\rho^{\Ztwo}(Z) \), we obtain the \( \Ztzt \)-supersymmetry model in~\cite[Section~III]{BD2020a} (with \( m=1/2 \)).

The above example demonstrates that the loop module provides a natural method for constructing the Hilbert space of a \( \Ztzt \)-graded supersymmetry model from a \( \Ztwo \)-graded supersymmetry model.
The loop module can be applied to other colour supersymmetric models constructed in a similar way.

\subsection{Loop module and irreducibility}
Under the assumptions outlined at the beginning of this section,
and for \( \R \) a \( \Gamma \)-graded associative algebra,
\cite{MZ2018} showed that every \( \Gamma \)-graded irreducible module can be obtained from the non-graded ones via the loop module:
\begin{thm}[{Mazorchuk--Zhao~\cite[Theorem~24]{MZ2018}}]\label{thm:MZ2018}
	Let \( W \) be a \( \Gamma \)-graded irreducible \( \R \)-module.
	There exists a subgroup \( H\leq \Gamma \) and an ungraded irreducible \( \R \)-module \( V \) such that \( W \) is isomorphic to \( \loopmodule[\Gamma]{V}{H} \) as a \( \Gamma \)-graded module.
\end{thm}
Note that~\cite{MZ2018} only proved the above result for graded Lie algebras,
but the same proof works for modules over any associative algebra.
A similar result~\cite[Theorem~31]{MZ2018} is available if \( \Gamma \) is an infinite group,
but we will restrict our attention to the finite case in this paper.

The above theorem reduces the study of \( \Gamma \)-graded irreducible modules to the study of ungraded modules.
In combination with discolouration,
this allows us to derive all such \( \Gamma \)-graded irreducible colour algebra representations from the \( \Ztwo \)-graded irreducible superalgebra representations.

However, one difficulty with applying \Cref{thm:MZ2018} in practice is that the choice of \( H \) in the theorem is not necessarily unique
i.e.\ potentially \( W\iso\loopmodule[\Gamma]{V}{H_1}\iso\loopmodule[\Gamma]{V}{H_2} \) for different subgroups \( H_1,H_2\leq \Gamma \).
Determining when two loop modules are isomorphic is an open problem~\cite[Problem~36]{MZ2018}.

Building on the work of~\cite{MZ2018}, it was shown that the loop module construction defines a functor and can be extended to an equivalence between two categories (see~\cite{EK2017} for details).
Moreover, using Brauer invariants, it was shown that the loop module defines a bijection between equivalence classes of \( \Gamma/H \)-graded irreducible and \( \Gamma \)-graded irreducible modules~\cite[Remark~7.2]{EK2017}.

\section{Iterated loop module bijection}\label{sec:loopirrep}
In this section, we will 
present a bijection between equivalence classes of \( \Gamma/H \)- and \( \Gamma \)-graded irreducible representations, similar to that of~\cite{EK2017}.
However, instead of using Brauer invariants, we will provide an alternate iterative approach.
Given a finite abelian grading group \( \Gamma \), we suggest using the Jordan--H\"older decomposition
\begin{equation}\label{eq:jordanholder}
	\Gamma\rhd N_{\ell} \rhd N_{\ell-1} \rhd \cdots \rhd N_1\rhd \{0\}
\end{equation}
and looking at submodules of the following iterated loop module
\begin{equation}\label{eq:iteratedloop}
	\loopmodule[\Gamma]{\cdots\loopmodule[\Gamma/N_{\ell-1}]{\loopmodule[\Gamma/N_\ell]{V}{\Gamma/N_\ell}}{N_\ell/N_{\ell-1}}\cdots}{N_1/\{0\}}.
\end{equation}
(Note that we are making use of the isomorphism \( \Gamma/N_{\ell}\iso (\Gamma/N_{\ell-1})/(N_\ell/N_{\ell-1}) \) (for instance)).
Restricting the loop module construction to cyclic groups of prime order greatly simplifies the possible options.

Our strategy is as follows:
we first prove that every finite-dimensional \( \Gamma \)-graded irreducible module appears as a quotient of a loop module (\Cref{thm:coarseirrep}).
The loop module is completely reducible both as a \( \Gamma/H \)- and \( \Gamma \)-graded module (\Cref{thm:VHdirectsum,thm:VHsemisimpleGmodule} respectively) and we explicitly provide the direct sum decompositions.
Comparing these direct sum decompositions limits the possible \( \Gamma \)-graded irreducible modules to either \( \Gamma/H \)-graded irreducible modules or their loop modules.
We use these results to construct the bijection between \( \Gamma/H \)- and \( \Gamma \)-graded modules (\Cref{thm:bijection}) up to equivalence defined by the irreducible direct summands.

\subsection{Quotient modules of the loop module}\label{sec:loopmodulequotient}
As in \Cref{sec:loopmodule},
let \( \R \) be a \( \Gamma \)-graded associative algebra
and \( V \) be a finite-dimensional \( \Gamma/H \)-graded irreducible \( \R \)-module.

In general, the loop module \( \loopmodule[\Gamma]{V}{H} \)
will not be \( \Gamma \)-graded irreducible.
However, to construct such a \( \Gamma \)-graded irreducible \( \R \)-module, we can simply quotient out by a maximal \( \Gamma \)-graded submodule.
Since \( \loopmodule[\Gamma]{V}{H} \) may have many maximal \( \Gamma \)-graded submodules, this procedure may yield many nonisomorphic \( \Gamma \)-graded modules for each maximal \( \Gamma \)-graded submodule.
Regardless, we will show that every \( \Gamma \)-graded irreducible \( \R \)-module can be constructed in this way.

\begin{lem}\label{thm:coarseirrep}
	Let \( \R=\bigoplus_{\gamma\in\Gamma}\R_{\gamma} \) be an associative algebra.
	Let \( W = \bigoplus_{\gamma\in\Gamma}W_{\gamma} \) be a finite-dimensional \( \Gamma \)-graded irreducible \( \R \)-module.
	Then, for any \( H\leq \Gamma \), there exists a \( \Gamma/H \)-graded irreducible \( \R \)-module \( V \)
	such that \( W \) is a graded quotient of \( \loopmodule[\Gamma]{V}{H} \).
\end{lem}

\begin{proof}
	Note that \( W \) has a natural \( \Gamma/H \)-grading given by
	\[
		\bigoplus_{\Lambda\in\Gamma/H}\left(\bigoplus_{\gamma\in\Lambda}W_{\gamma}\right).
	\]
	Let \( V \) be a \( \Gamma/H \)-graded irreducible submodule of \( W \).
	Such a \( V \) exists by strong induction on the dimension of \( W \).
	(See also~\cite[Lemma~23]{MZ2018}.)

	Take \( v\in V_\Lambda \) and write
	\begin{equation}\label{eq:coarseirreplincomb}
		v = \sum_{\gamma\in\Lambda}w_\gamma
	\end{equation}
	for \( w_\gamma\in W_\gamma \).
	For \( \gamma\in\Lambda \), let \( \pi_\gamma\colon V_\Lambda\to W_\gamma \)
	be the projection of \( V_\Lambda \) onto \( W_\gamma \),
	\[
		\pi_\gamma(v) = w_\gamma.
	\]
	Note that \( \pi_{\gamma} \) is well-defined because the linear combination in~\eqref{eq:coarseirreplincomb} is unique due to the direct sum \( V_\Lambda \subseteq \bigoplus_{\gamma\in\Lambda} W_{\gamma} \).
	Additionally, it is easy to see that \( \pi_{\gamma} \) is linear.

	Note that,
	for \( x\in\R_{\alpha} \),
	we have \( xv = \sum_{\gamma\in\Lambda}x\pi_{\gamma}(v) \).
	Since \( x\pi_{\gamma}(v)\in W_{\alpha+\gamma} \) and 
	such a linear combination of homogeneous elements is unique, we have that \( \pi_{\alpha+\gamma}(xv) = x\pi_{\gamma}(v) \) by the definition of \( \pi_{\alpha+\gamma} \).

	Define a map \( \Phi\colon \loopmodule[\Gamma]{V}{H} \to W \) by
	\[
		\Phi(v\otimes e_{\gamma}) = \pi_{\gamma}(v).
	\]
	Then
	\[
		x\Phi(v\otimes e_{\gamma}) = x\pi_{\gamma}(v) = \pi_{\alpha+\gamma}(xv) = \Phi(x(v\otimes e_{\gamma}))
	\]
	so \( \Phi \) is an intertwiner.

	Therefore, 
	by the First Isomorphism Theorem,
	\( \im \Phi \iso \loopmodule[\Gamma]{V}{H}/\ker \Phi \).
	But \( \im \Phi \) is a \( \Gamma \)-graded submodule of \( W \),
	so \( \im \Phi=0 \) or \( \im \Phi=W \) by irreducibility of \( W \).
	However,
	since \( V\neq 0 \),
	and \( \Phi \) is defined using projections,
	we clearly have \( \im \Phi \neq 0 \).
	We conclude that \( \im \Phi=W \);
	hence \( W\iso \loopmodule[\Gamma]{V}{H}/\ker \Phi \) as claimed.
\end{proof}

\Cref{thm:coarseirrep} provides us with a general method of obtaining the \( \Gamma \)-graded irreducible modules from the \( \Gamma/H \)-graded ones; namely find all maximal \( \Gamma \)-graded submodules of \( \loopmodule[\Gamma]{V}{H} \) and take quotients.
Note that we can first fix \( H \) and then find all \( \Gamma \)-graded irreducible modules.
Compare this with~\Cref{thm:MZ2018}, where we must choose a different \( H \) for each \( \Gamma \)-graded module.

\subsection{\texorpdfstring{\( \Gamma/H \)}{Gamma/H}-grading: loop module is completely reducible}\label{sec:G/Hsemisimple}
The loop module is a completely reducible \( \Gamma/H \)-graded module.
It is shown in~\cite{EK2015,MZ2018} that 
the irreducible direct summands are isomorphic to \( V \) as vector spaces, but the \( \R \)-action is twisted by an element of the character group \( \pdual{\Gamma} \).
These direct summands will restrict which quotients are possible,
and hence which \( \Gamma \)-graded irreducible representations are possible in \Cref{thm:coarseirrep}.
We will now briefly recount these results.

Let \( \pdual{\Gamma} \) be the character group:
the group of all homomorphisms \( \pdual{\Gamma}\to\units{\FF} \) with pointwise multiplication the group operation.
Our earlier assumption that \( \FF \) is algebraically closed 
ensures that \( \pdual{\Gamma}\iso\Gamma \).
Let \( \orthog{H} = \{f\in\pdual{\Gamma}\mid f(\eta) = 1~\text{for all}~\eta\in H\} \).
Note that since \( \Gamma \) is finite, \( \card*{\pdualtwist} = \card{H} \).

Let \( V = \bigoplus_{\Lambda\in\Gamma/H} V_{\Lambda} \) be a
\( \Gamma/H \)-graded \( \R \)-module with representation \( \rho\colon\R\to\End(V) \).
For each \( f\in \pdual{\Gamma} \),
we define the twisted module \( V^f \) to be equal to \( V \) as a vector space,
but with the following twisted representation \( \rho^f\colon\R\to\End(V^f) \)
\[
	\rho^f(x)v = f(\alpha)\rho(x)v \qquad \text{for}~x\in\R_{\alpha},\,v\in V.
\]
Note that \( \rho^f \) depends on the choice of \( f\in \pdual{\Gamma} \),
but if \( f \) and \( f' \) belong to the same coset in \( \pdualtwist \)
then \( \rho^f \) and \( \rho^{f'} \) are isomorphic
(that is, \( f\orthog{H} \) defines \( \rho^f \) up to isomorphism).
Note also that the definition of \( \rho^f \) depends only on the grading of \( \R \)
and \emph{not} on the grading of \( V \)
(indeed, the definition makes sense even if \( V \) is ungraded).

%

\begin{rmk}\label{rmk:smashproduct}
	\( \Gamma \)-graded modules over \( \R \) are equivalent to modules over the smash product \( \R\#\FF\pdual{\Gamma} \) (see~\cite{EK2015}); that is, \( \Gamma \)-graded is equivalent to \( \pdual{\Gamma} \)-invariant.
	Furthermore, \( \Gamma/H \)-graded modules are equivalent to modules over the smash product
	\[
		\R\#\FF\orthog{H} \subset \R\#\FF\pdual{\Gamma}.
	\]
	With this set-up, it is natural to consider the induced module \( \Ind_{\orthog{H}}^{\pdual{\Gamma}} V = V\otimes_{\FF\orthog{H}} \FF\pdual{\Gamma} \) which can be given a \( \Gamma \)-grading.
	(Under our assumptions, this induced module is isomorphic to the loop module \( \loopmodule[\Gamma]{V}{H} \) \cite[Proposition~3.8]{EK2017}.)
	With this interpretation, we could reframe the results of this section (\Cref{sec:loopirrep}) in terms of Clifford theory.
	For the relevant Clifford theory in terms of gradings, we refer the reader to~\cite{EK2015}.
\end{rmk}

For later use, we prove the following proposition.

\begin{prop}\label{prop:V-=Vifgrade}
	Let \( H\leq \Gamma \) be cyclic of order \( n \).
	The irreducible \( \Gamma/H \)-graded \( \R \)-module \( V \) can be given a \( \Gamma \)-grading that refines the \( \Gamma/H \)-grading
	if and only if \( V^f\iso V \) as \( \Gamma/H \)-graded modules for every \( f\in\pdual{\Gamma} \).
\end{prop}

Note that the above proposition does not hold in general:
if \( H \) were an arbitrary group
then we would also need graded Schur indices~\cite[Definition 4]{EK2015}
to characterise when \( V \) is \( \Gamma \)-gradable.

\begin{proof}
	For the forward direction,
	assume that \( V = \bigoplus_{\gamma\in \Gamma} W_{\gamma} \) is a \( \Gamma \)-grading.
	Then it is easily verified that the linear map defined by
	\( v\mapsto f(\xi)v \)
	(for \( \xi \)-graded \( v \)),
	is an isomorphism \( V\to V^f \).

	For the reverse direction,
	choose \( f\in \pdual{\Gamma} \)
	such that \( \restrict{f}{H} \) is a generator for \( \pdual{H} \).
	Let \( \Psi\colon V\to V^f \) be an isomorphism.
	Define \( \Phi \colon V \to V^{f^n}  \) by
	\[
		\Phi(v) = (f(\gamma))^nv
	\]
	where \( v \) has degree \( (\gamma+H)\in \Gamma/H \).
	This map \( \Phi \)
	is well-defined
	since \( f^n\in \orthog{H} \)
	(recall \( \card*{\pdualtwist} = n \)).
	As in the forward direction of the proof,
	\( \Phi \) is an isomorphism.

	By the graded Schur's lemma \cite[Theorem~14]{MZ2018},
	we have that \( \Psi^n = \lambda\Phi \) for some \( \lambda\in\units{\FF} \).
	Without loss of generality, we can rescale \( \Psi \) so that \( \Psi^n = \Phi \).

	For each \( \gamma\in \Gamma \), define
	\[
		P_\gamma = \frac{1}{n}\sum_{j=1}^n (f(\gamma))^{-j}\Psi^{j}.
	\]
	
	Let \( h\in H \).
	Restrict the domain
	to \( V_{\gamma+H} \):
	using the fact that \( \Psi^{n+1} = \Phi\Psi = (f(\gamma))^n\Psi \) on \( V_{\gamma+H} \)
	we find that \( \Psi P_\gamma = f(\gamma) P_\gamma \).
	Therefore,
	\begin{align*}
		P_{\gamma+h}P_{\gamma}
		&= \frac{1}{n}\sum_{j=1}^n (f(\gamma+h))^{-j}(f(\gamma))^j P_\gamma
		= \left(\frac{1}{n}\sum_{j=1}^n (f(h))^{-j}\right) P_\gamma.
	\end{align*}
	We know that \( f(h) \) is an \( n \)-th root of unity.
	Moreover, since
	\( \restrict{f}{H} \) generates \( \pdual{H} \),
	we find that \( f(\eta) \) is a primitive \( n \)-th root of unity
	for any generator \( \eta \) of \( H \),
	and hence \( f(h) = 1 \) if and only if \( h=0 \).
	Consequently,
	\begin{equation}\label{eq:prodproj}
		P_{\gamma+h}P_\gamma
		=
		\begin{cases}
			P_\gamma & \text{if}~h=0\\
			0 & \text{if}~h\neq 0.
		\end{cases}
	\end{equation}
	In particular, notice that \( P_\gamma \) is a projection operator.

	Additionally,
	\[
		\sum_{h\in H} P_{\gamma+h}
		= \frac{1}{n}\sum_{j=1}^n \sum_{h\in H}(f(\gamma+h))^{-j}\Psi^j
		= \frac{1}{n}\sum_{j=1}^n \left(\sum_{h\in H}(f(h))^{-j}\right)(f(\gamma))^{-j}\Psi^j.
	\]
	Since \( \restrict{f}{H} \) generates \( \pdual{H} \),
	we know that \( \restrict{f^j}{H}\neq 1 \) for \( 0 < j < n \).
	In particular,
	\( f^j(\eta) \neq 1 \) for any generator \( \eta \) of \( H \).
	Since, moreover, \( f^j(\eta) \) is an \( n \)-th root of unity,
	we find that \( \sum_{h\in H}(f(h))^{-j} = 0 \) for \( 0 < j < n \).
	Therefore,
	\begin{equation}\label{eq:sumproj}
		\sum_{h\in H} P_{\gamma+h}
		= (f(\gamma))^{-n}\Psi^n = \Id.
	\end{equation}

	Setting, \( W_{\gamma} = P_\gamma(V_{\gamma+H}) \),
	the equations~\eqref{eq:prodproj} and~\eqref{eq:sumproj} imply that
	\[
		V_{\gamma+H} = \bigoplus_{h\in H} W_{\gamma +h}
	\]
	and hence \( V = \bigoplus_{\gamma\in \Gamma} W_\gamma \).
	We claim that this is a \( \Gamma \)-grading for \( V \).
	Indeed, let \( \rho\colon \R\to\End(V) \) be the representation
	and take \( x\in\R_\alpha \)
	and \( v\in W_\gamma \); then
	\begin{align*}
		P_{\alpha+\gamma}(\rho(x)v)
		&= 
		\frac{1}{n}\sum_{j=1}^n (f(\alpha+\gamma))^{-j}\Psi^{j}(\rho(x)v)\\
		&=\frac{1}{n}\sum_{j=1}^n (f(\alpha+\gamma))^{-j}\rho^{f^j}(x)\Psi^{j}(v)\\
		&=\frac{1}{n}\sum_{j=1}^n (f(\alpha+\gamma))^{-j}(f(\alpha))^j\rho(x)\Psi^{j}(v)\\
		&= \rho(x)P_\gamma v\\
		&= \rho(x)v.
	\end{align*}
	That is \( \rho(x)v \in W_{\alpha+\gamma} \)
	as required.
\end{proof}


\begin{rmk}\label{rmk:V-=Vifgradeproj}
	Given an isomorphism \( \Psi\colon V\to V^f \),
	we can use the projection operators
	\[
		P_\gamma = \frac{1}{n}\sum_{j=1}^n (f(\gamma))^{-j}\Psi^{j}
	\]
	used in the above proof
	to explicitly calculate
	the \( \Gamma \)-grading
	of an irreducible \( \Gamma/H \)-graded \( \R \)-module \( V \).
	%
\end{rmk}

Using the twisted modules \( V^f \),
we can give a direct sum decomposition of the loop module
as a \( \Gamma/H \)-graded module:
\begin{lem}[{Elduque--Kochetov~\cite[section~3.3]{EK2015}; see also, Mazorchuk--Zhao~\cite[Lemma~20]{MZ2018}}]\label{thm:VHdirectsum}
	Let \( H
	\) be a finite
	group and \( V = \bigoplus_{\Lambda\in\Gamma/H}V_\Lambda \) be a \( \Gamma/H \)-graded module of an associative algebra \( \R \) over \( \FF \).
	Assume that \( \FF \) has a primitive root of unity of order \( \card{\Gamma} \).
	As a \( \Gamma/H \)-graded module, the loop module \( \loopmodule[\Gamma]{V}{H} \) is isomorphic to
	\[
		\bigoplus_{f\orthog{H}\in \pdualtwist}V^f.
	\]
	(The direct sum ranges over \( f \) in some transversal of \( \pdualtwist \)
	so that there are \( \card*{\pdualtwist}=\card*{H} \) direct summands.)
\end{lem}

If we assume that \( V \) is irreducible,
then \Cref{thm:VHdirectsum} (together with~\Cref{prop:twistproperties}~\ref{prop:V-irred} below) shows that \( \loopmodule[\Gamma]{V}{H} \) is completely reducible as a \( \Gamma/H \)-graded module and gives a direct sum decomposition.
In combination with \Cref{thm:coarseirrep},
this restricts the possible structures of the \( \Gamma \)-graded irreducible modules.

We now prove two more important properties of \( V^f \).

\begin{prop}\label{prop:twistproperties}
	The following two properties hold:
	\begin{enumerate}[label=(\roman*)]
		\item\label{prop:V-irred} If \( V \) is \( \Gamma/H \)-graded irreducible then \( V^f \) is \( \Gamma/H \)-graded irreducible.
		\item\label{prop:VH-=V-H} There is an equality of \( \R \)-modules: \( (\loopmodule[\Gamma]{V}{H})^f = \loopmodule[\Gamma]{V^f}{H} \).
	\end{enumerate}
\end{prop}

\begin{rmk}
	Property~\ref{prop:V-irred} has previously been proven in~\cite[Section~3.1]{EK2015}.
	A result analogous to property~\ref{prop:VH-=V-H} for induced modules (see \Cref{rmk:smashproduct}) follows immediately from the definition.
\end{rmk}

\begin{proof}
	Observe that \( \rho(\R)U\subseteq U \) if and only if \( \rho^f(\R)U\subseteq U \) for any submodule \( U \) of \( V \) or \( V^f \).
	Property~\ref{prop:V-irred} follows.

	For~\ref{prop:VH-=V-H},
	\( (\loopmodule[\Gamma]{V}{H})^f \) equals \( \loopmodule[\Gamma]{V^f}{H} \) as a vector space
	and both of these modules have the same representation:
	for \( x\in\R_{\alpha} \) and \( v\in V_{\gamma+H} \),
	\[
		x\cdot(v\otimes e_{\gamma}) = 
		f(\alpha) \rho(x)v\otimes e_\gamma.
		\qedhere
	\]
\end{proof}

It is clear from the definition and the above propositions that \( V^f \) has a very similar module structure to \( V \).
This will be useful when examining the structure of \( \loopmodule[\Gamma]{V}{H} \).


\subsection{\texorpdfstring{\( \Gamma \)}{Gamma}-grading: loop module is completely reducible}\label{sec:Gsemisimple}
We will now show that the loop module is completely reducible as a \( \Gamma \)-graded module.
Similar to the \( \Gamma/H \)-graded case, the \( \Gamma \)-graded irreducible direct summands all have a similar structure: they are all (what we will call) \( H \)-parity shifts of the same \( \Gamma \)-graded irreducible module.
A \( H \)-parity shift is obtained by simply applying a circular shift to each of the sectors.

Comparing the \( \Gamma \)- and \( \Gamma/H \)-graded irreducible direct summands will give us more information about the structure of the loop module.
This will allow us to conclude that the only possible \( \Gamma \)-graded direct summands of \( \loopmodule[\Gamma]{V}{H} \) are \( V \) with a \( \Gamma \)-grading or \( \loopmodule[\Gamma]{V}{H} \) itself.

\subsubsection{Direct sum decomposition}
Take some \( h\in H \) and a \( \Gamma \)-graded module \( W \).
We can define the \dfn{\( H \)-parity shift of \( W \) by \( h \)}
as a new \( \Gamma \)-graded module \( W^{+h} = \bigoplus_{\gamma\in\Gamma}W^{+h}_\gamma \),
where \( W^{+h}_\gamma = W_{\gamma+h} \).
Obviously, a \( H \)-parity shift of \( W \) is still a \( \Gamma \)-graded \( \R \)-module.

Following this definition, we make four immediate observations:
\begin{prop}\label{prop:Hparityshiftimmediateobservations}
	The following properties hold:
	\begin{enumerate}[label=(\roman*)]
		\item \label{prop:Linoutshift}
			For any \( \xi\in \Gamma \),
			we have \( \loopmodule[\Gamma]{V^{+(\xi+H)}}{H} \iso (\loopmodule[\Gamma]{V}{H})^{+\xi} \)
			as \( \Gamma \)-graded modules.
		\item\label{prop:L+hisom} 
			For any \( h\in H \),
			the \( H \)-parity shift
			\( (\loopmodule[\Gamma]{V}{H})^{+h} \) is isomorphic to \( \loopmodule[\Gamma]{V}{H} \).
		\item\label{prop:W+hirred} If \( W \) is \( \Gamma \)-graded irreducible then \( W^{+h} \) is \( \Gamma \)-graded irreducible.
		\item\label{prop:WW+hisom} The \( \R \)-modules \( W \) and \( W^{+h} \) are isomorphic as \( \Gamma/H \)-graded modules.
	\end{enumerate}
\end{prop}
\begin{proof}

	For~\ref{prop:Linoutshift},
	the \( \gamma \)-sector of \( \loopmodule[\Gamma]{V^{+(\xi+H)}}{H} \)
	is \( V_{\gamma+\xi+H}\otimes \FF e_\gamma \)
	and the \( \gamma \)-sector of \( (\loopmodule[\Gamma]{V}{H})^{+\xi} \)
	is \( V_{\gamma+\xi+H}\otimes \FF e_{\gamma+\xi} \).
	Thus, we can easily verify
	the linear map \( \loopmodule[\Gamma]{V^{+(\xi+H)}}{H} \to (\loopmodule[\Gamma]{V}{H})^{+\xi} \) defined by \( v\otimes e_\gamma\mapsto v\otimes e_{\gamma+\xi} \)
	is a \( \Gamma \)-graded isomorphism.

	Result~\ref{prop:L+hisom}
	is a special case of~\ref{prop:Linoutshift}.
	The proofs of~\ref{prop:W+hirred}~and~\ref{prop:WW+hisom} are immediate.
\end{proof}

\begin{prop}\label{lem:Hparityshiftsubmodule}
	Let \( V \) be a \( \Gamma/H \)-graded module and \( W \) be a \( \Gamma \)-graded submodule of \( \loopmodule[\Gamma]{V}{H} \).
	Then \( W^{+h} \) can be embedded in \( \loopmodule[\Gamma]{V}{H} \).
\end{prop}
\begin{proof}
	Define a map \( \Phi_h\colon (\loopmodule[\Gamma]{V}{H})^{+h} \to \loopmodule[\Gamma]{V}{H} \) by
	\[
		\Phi_h(v\otimes e_{\gamma+h}) =  v\otimes e_{\gamma}
	\]
	for \( v\otimes e_{\gamma+h}\in (\loopmodule[\Gamma]{V}{H})^{+h}_\gamma \) and extending to inhomogeneous elements by linearity.
	The restriction \( \restrict{\Phi_h}{W} \)
	is then an embedding of \( W^{+h} \) into \( \loopmodule[\Gamma]{V}{H} \).
\end{proof}

\begin{lem}\label{thm:VHsemisimpleGmodule}
	Let \( H\leq \Gamma \) and \( V = \bigoplus_{\Lambda\in\Gamma/H} V_{\Lambda} \) be a finite-dimensional \( \Gamma/H \)-graded irreducible module of an associative algebra \( \R \) over \( \FF \).
	Then, the loop module \( \loopmodule[\Gamma]{V}{H} \) is completely reducible as a \( \Gamma \)-graded module with
	\[
		\loopmodule[\Gamma]{V}{H} \iso \bigoplus_{h \in K} W^{+h}
	\]
	for any \( \Gamma \)-graded irreducible module \( W\subseteq \loopmodule[\Gamma]{V}{H} \)
	and some subset \( K \subseteq H \).
\end{lem}
\begin{proof}
	Take any \( \Gamma \)-graded irreducible module \( W\subseteq \loopmodule[\Gamma]{V}{H} \).
	Such a module exists because \( \loopmodule[\Gamma]{V}{H} \) is finite-dimensional.

	Define a map \( \Pi\colon \loopmodule[\Gamma]{V}{H} \to V \) by 
	\[
		\Pi(v\otimes e_{\gamma}) = v
	\]
	for \( v\otimes e_\gamma \in \loopmodule[\Gamma]{V}{H}_\gamma \)
	and extending to inhomogeneous elements by linearity.
	We claim that \( \Pi \) is a \( \Gamma/H \)-graded intertwiner.
	That \( \Pi \) is linear and homogeneous of degree \( 0 \) is obvious.
	For \( x\in\R_\alpha \), we have
	\[
		\Pi(x(v\otimes e_{\gamma}))
		=\Pi(xv\otimes e_{\alpha+\gamma})
		=xv 
		=x\Pi(v\otimes e_{\gamma})
	\]
	which proves the claim.

	The fact that \( \Pi \) is an intertwiner implies that \( \Pi(W) \) is a \( \Gamma/H \)-graded submodule of \( V \).
	But since \( W\neq 0 \), 
	we must have that \( \Pi(W)\neq 0 \).
	Since \( V \) is \( \Gamma/H \)-graded irreducible, we must have that \( \Pi(W) = V \).

	Now, consider the module \( \Sigma_{h\in H} \Phi_h(W^{+h}) \subseteq \loopmodule[\Gamma]{V}{H} \)
	(where \( \Phi_h \) embeds \( W^{+h} \) inside \( \loopmodule[\Gamma]{V}{H} \) as in the proof of \Cref{lem:Hparityshiftsubmodule}).
	In particular, consider the \( \gamma \)-sector
	\( \sum_{h\in H} \Phi_h(W^{+h}_\gamma) \) for some \( \gamma\in\Gamma \).
	Applying the map \( \Pi \), 
	\[
		\Pi\left(\sum_{h\in H}\Phi_h(W^{+h}_\gamma)\right)
		= \Pi\left(\bigoplus_{h\in H} \Phi_h(W_{\gamma+h})\right) = V_{\gamma+H}
	\]
	using the fact \( \Pi(W) = V \).
	By examining the definition of \( \Pi \), we find that \( \sum_{h\in H}\Phi_h(W^{+h}_\gamma) = V_{\gamma+H}\otimes e_\gamma = \loopmodule[\Gamma]{V}{H}_\gamma \).
	Therefore, \( \sum_{h\in H} \Phi_h(W^{+h}) = \loopmodule[\Gamma]{V}{H} \).

	All that remains to show is that we can take a subset of \( H \) to obtain a direct sum.
	But this follows from \Cref{prop:sumsimpledirect}.
\end{proof}

\begin{thm}\label{cor:VorVHirred}
	Let \( H\leq \Gamma \) be a cyclic group of prime order \( p \)
	and let \( V \) be \( \Gamma/H \)-graded irreducible.
	Every \( \Gamma \)-graded irreducible submodule of \( \loopmodule[\Gamma]{V}{H} \) is isomorphic as a \( \Gamma/H \)-graded submodule to either \( \loopmodule[\Gamma]{V}{H} \) or \( V \).
\end{thm}
\begin{proof}
	Take \( W \) a \( \Gamma \)-graded irreducible submodule of \( \loopmodule[\Gamma]{V}{H} \).
	By \Cref{thm:VHsemisimpleGmodule},
	\[
		\loopmodule[\Gamma]{V}{H} \iso \bigoplus_{h\in K} W^{+h}.
	\]

	Since \( \loopmodule[\Gamma]{V}{H} \) is a completely reducible \( \Gamma/H \)-graded module, we know that \( W \) is a completely reducible \( \Gamma/H \)-graded module with direct summands contained in those of \( \loopmodule[\Gamma]{V}{H} \).
	By \Cref{thm:VHdirectsum}, \( W \) is a direct sum of character twists \( V^f \) (\( f\in\pdual{\Gamma} \)).
	Let \( n \) be the number of \( \Gamma/H \)-graded direct summands in \( W \).

	Then, the number of direct summands in \( \loopmodule[\Gamma]{V}{H} \) as a \( \Gamma/H \)-graded module is equal to \( n\card{K} \).
	But, from \Cref{thm:VHdirectsum}, the number of direct summands is equal to \( \card*{\pdualtwist} = \card{H} = p \).
	Therefore, \( n=1 \) or \( n=p \). 
	If \( n=p \) then \( W\iso \loopmodule[\Gamma]{V}{H} \).
	If \( n=1 \) then \( W \) is isomorphic to \( V^f \) for some \( f\in\pdual{\Gamma} \).
	But then \( W \) gives \( V^f \) a \( \Gamma \)-grading,
	so \( W\iso V^f \iso V \) by \Cref{prop:V-=Vifgrade}.
\end{proof}

In combination with \Cref{thm:coarseirrep},
\Cref{cor:VorVHirred} tells us that every \( \Gamma \)-graded irreducible module appears as either a \( \Gamma/H \)-graded irreducible module or the loop module of one.
This result will greatly aid in the search for graded irreducible representations.
Noting that \( V \) is isomorphic to \( \loopmodule[\Gamma/H]{V}{\{0\}} \),
\Cref{cor:VorVHirred} is a similar result as~\cite[Theorem~24]{MZ2018} (restated in~\Cref{thm:MZ2018}):
every \( \Gamma \)-graded module is a loop module.
However, there are two important distinctions between \Cref{cor:VorVHirred} and~\cite[Theorem~24]{MZ2018}:
restricting \( H \) to a cyclic group of prime order greatly restricts the possible loop modules,
and \( H \) does not rely on the choice of \( V \).

One difficulty remains: finding a \( \Gamma \)-grade for a \( \Gamma/H \)-graded module can be difficult in general.
For some results about the gradability of modules see~\cite{EK2015}.
To aid with this, we prove a uniqueness result for the \( \Gamma \)-grade.

\subsubsection{Uniqueness of the \texorpdfstring{\( \Gamma \)}{Gamma}-grade}
It is possible for two modules to be non-isomorphic as \( \Gamma \)-graded modules but isomorphic as \( \Gamma/H \)-graded modules.
This is important to keep in mind when trying to find a \( \Gamma \)-grading for a \( \Gamma/H \)-graded irreducible module; there might be multiple \emph{different} \( \Gamma \)-gradings, all of which must be considered.
However, the following proposition gives a uniqueness result for the \( \Gamma \)-grade.

\begin{prop}\label{prop:gradinguniqueness}
	If \( V \) is a finite-dimensional \( \Gamma/H \)-graded irreducible \( \R \)-module, then any two \( \Gamma \)-gradings for \( V \) are equivalent up to \( H \)-parity shift.
\end{prop}

More precisely, this proposition tells us that if
we have two \( \R \)-modules 
\( V = \bigoplus_{\gamma\in\Gamma} V_{\gamma} \) and \( \widetilde{V} = \bigoplus_{\gamma\in\Gamma} \widetilde{V}_{\gamma} \)
with \( V=\widetilde{V} \) as \( \Gamma/H \)-graded modules
(i.e.\ for each \( \Lambda\in \Gamma/H \) the \( \Lambda \)-sector of \( V=\widetilde{V} \) is \( \bigoplus_{\gamma\in\Lambda}V_{\gamma}=\bigoplus_{\gamma\in\Lambda}\widetilde{V}_{\gamma} \))
then there exists \( \eta\in H \) such that \( V^{+\eta} \iso \widetilde{V} \)
as \( \Gamma \)-graded modules.

\begin{proof}
	Both \( V = \bigoplus_{\gamma\in\Gamma} V_{\gamma} \) and \( \widetilde{V} = \bigoplus_{\gamma\in\Gamma} \widetilde{V}_{\gamma} \) appear as \( \Gamma \)-graded irreducible submodules of \( \loopmodule[\Gamma]{V}{H} \) by \Cref{prop:gradesubrep}.
	Since all \( \Gamma \)-graded irreducible direct summands of \( \loopmodule[\Gamma]{V}{H} \) are \( H \)-parity shifts of \( V \) by \Cref{thm:VHsemisimpleGmodule},
	the result follows.
\end{proof}

%
%

\subsubsection{Going backwards: from \( \Gamma \)-graded to \( \Gamma/H \)-graded irreducible}

So far, we have been focussing on constructing \( \Gamma \)-graded irreducible modules from \( \Gamma/H \)-graded modules.
But what if we want to go in the reverse direction?

The loop module of every \( \Gamma/H \)-graded irreducible module \( V \) exists
and, by \Cref{cor:VorVHirred}, every \( \Gamma \)-graded submodule of \( \loopmodule[\Gamma]{V}{H} \) contains \( V \) as a \( \Gamma/H \)-graded submodule.
That is, every \( \Gamma/H \)-graded irreducible module appears as the submodule of a \( \Gamma \)-graded one.

To aid in finding \( \Gamma/H \)-graded submodules,
we can prove the analogue of \Cref{prop:V-=Vifgrade} for \( H \)-parity shifts:
\begin{prop}\label{prop:V+h=Vifred}
	Let \( W = \bigoplus_{\gamma\in\Gamma} W_\gamma \) be a \( \Gamma \)-graded irreducible module
	and \( H \) be cyclic of prime order.
	The module \( W \) is \( \Gamma/H \)-graded reducible if and only if all \( H \)-parity shifts of \( W \) are \( \Gamma \)-graded isomorphic.
\end{prop}
\begin{proof}
	For the forward direction, assume \( W \) is \( \Gamma/H \)-graded reducible.
	By \Cref{thm:coarseirrep} and \Cref{cor:VorVHirred}, we must have that \( W \) is isomorphic to a loop module.
	But,
	by~\Cref{prop:Hparityshiftimmediateobservations}~\ref{prop:L+hisom},
	all \( H \)-parity shifts of a loop module are \( \Gamma \)-graded isomorphic,
	as required.

	For the reverse direction,
	suppose \( W \) is isomorphic to \( W^{+h} \).
	There exists some \( \Gamma \)-graded isomorphism \( \Psi\colon W\to W^{+h} \).
	This \( \Psi \) is a \( \Gamma/H \)-graded automorphism
	(but \emph{not} a \( \Gamma \)-graded automorphism).
	This automorphism \( \Psi \) cannot be a scalar multiple of the identity,
	since \( W_{\gamma}\to W^{+h}_\gamma = W_{\gamma+h} \) for every \( \gamma\in\Gamma \).
	Therefore, by the graded Schur's lemma \cite[Theorem~14]{MZ2018},
	\( W \) must be \( \Gamma/H \)-graded reducible.
\end{proof}

\begin{rmk}\label{rmk:V+h=Vifredproj}
	We can give a formula for \( \Gamma/H \)-graded submodules in the above proposition.
	If \( \Psi\colon W\to W^{+h} \) is an isomorphism such that \( \Psi^p = \Id \) as a linear map
	(provided \( h\neq 0 \)),
	then a projection onto an irreducible \( \Gamma/H \)-graded submodule is given by
	\[
		P = \frac{1}{p}\sum_{j=1}^p \Psi^j.
	\]
	Indeed, let \( \Phi_{-h}\colon \loopmodule[\Gamma]{V}{H}\to \loopmodule[\Gamma]{V}{H}^{+h} \) be the isomorphism defined by
	\( v\otimes e_\xi \mapsto v\otimes e_{\xi+h} \)
	(c.f.\ the proof of \Cref{lem:Hparityshiftsubmodule}).
	By the above proposition, \( W \) is reducible and hence \( W\iso\loopmodule[\Gamma]{V}{H} \) for some \( \Gamma/H \)-graded \( V \);
	identify \( W \) with \( \loopmodule[\Gamma]{V}{H} \).
	Then \( \Psi\circ \Phi^{-1}_{-h} \colon W\to W  \) is equal to \( \lambda\Id \) for some \( \lambda\in\FF \) by Schur's Lemma.
	That is, \( \Psi = \lambda\Phi_{-h} \).
	The assumption \( \Psi^p = \Id \) is equivalent to \( \lambda^p = 1 \).
	Then \( \lambda \) is a \( p \)-th root of unity;
	choose \( f\in\pdual{\Gamma} \) such that \( f(h) = \lambda^{-1} \)
	(such an \( f \) exists because \( \card*{\pdualtwist} = p \))
	and define a linear map \( \varphi\colon V^f \to P(\loopmodule[\Gamma]{V}{H}) \)
	by \( \varphi(v) = f(\gamma)^{-1}P(v\otimes e_\gamma) \)
	for \( v\in V^f_{\gamma+H} \).
	The fact that \( \Psi = \lambda\Phi_{-h} \) and the form of \( P \)
	both ensure that \( \varphi \) is well-defined.
	It is easily verified that \( \varphi \) is an isomorphism
	with inverse defined by restricting the map \( v\otimes e_\xi\mapsto f(\xi)v \).
\end{rmk}

\subsection{The bijection}\label{sec:bijection}

Let \( G \) be a finite group and \( \R=\bigoplus_{g\in G}\R_g \) be an associative algebra.
For any subgroup \( M \leq G \),
let \( \irrep{G,M} \) be the collection of equivalence classes of finite-dimensional \( G/M \)-graded irreducible \( \R \)-modules,
where two \( G/M \)-graded modules \( \widetilde{V} \) and \( V \) are equivalent 
if \( \widetilde{V} \iso (V^{+\Lambda})^f \)
for some \( \Lambda\in G/M \) and some \( f\in\pdual{G} \).
That this forms an equivalence relation is easily verified.

\begin{thm}\label{thm:bijection}
	Let \( G \) be a finite abelian group
	and \( \R \) a \( G \)-graded associative algebra.
	Let \( N \leq M \leq G \)
	and 
	set \( \Gamma = G/N \)
	and \( H = M/N \).
	Assume \( H \) is a cyclic group of prime order.
	Define a function \( \bij\colon\irrep{G,M} \to \irrep{G,N} \) by
	\begin{equation}\label{eq:bijection}
		\bij([V]) = 
		\begin{cases}
			[V_\Gamma] & \text{if \( V \) is \( \Gamma \)-gradable}\\
			[\loopmodule[\Gamma]{V}{H}] & \text{otherwise}
		\end{cases}
	\end{equation}
	where \( V \), \( V_\Gamma \) and \( \loopmodule[\Gamma]{V}{H} \) are representatives of the equivalence classes,
	and \( V_{\Gamma} \) is \( \Gamma \) graded but isomorphic to \( V \) as a \( \Gamma/H \)-graded module.
	(To construct the loop module \( \loopmodule[\Gamma]{V}{H} \), we use the group isomorphism \( G/M\iso \Gamma/H \).)
	Then \( \bij \) is well-defined (does not depend on the choice of representative \( V \)) and bijective.
\end{thm}
\begin{proof}
	First, we show that \( \bij \) is well-defined.
	Choose two representatives \( \widetilde{V} \) and \( V \)
	of an equivalence class in \( \irrep{G,M} \),
	so that \( [\widetilde{V}] = [V] \).
	By definition, \( \widetilde{V} \) is isomorphic to \( (V^{+(\gamma+H)})^f \)
	for some \( \gamma+H\in\Gamma/H\iso G/M \)
	and some \( f\in\pdual{G} \).
	Note that the \( \Gamma \)-gradability is unchanged after twisting by \( f \) and after shifting by \( \gamma+H \).
	Therefore, \( \bij(\widetilde{V}) = \widetilde{V}_\Gamma \) if and only if \( \bij(V) = V_\Gamma \).
	There are two cases.
	\begin{itemize}
		\item Assume \( \bij([\widetilde{V}]) = [\widetilde{V}_\Gamma] \) and \( \bij([V]) = [V_\Gamma] \).
			As \( \Gamma/H \)-graded modules,
			\( \widetilde{V}_\Gamma\iso (V_\Gamma^{+\gamma})^f \).
			By~\Cref{prop:gradinguniqueness},
			the \( \Gamma \)-grading of \( \widetilde{V} \)
			is unique up to \( H \)-parity shift.
			Therefore, as \( \Gamma \)-graded modules, \( \widetilde{V}_\Gamma \iso ((V_\Gamma^{+\gamma})^f)^{+h} \iso (V_\Gamma^{+(\gamma+h)})^f \)
			for some \( h\in H \).
			That is, \( [\widetilde{V}_\Gamma] = [V_\Gamma] \)
			and \( \bij \) is well-defined in this case.
		\item Assume \( \bij([\widetilde{V}]) = [\loopmodule[\Gamma]{\widetilde{V}}{H}] \) and \( \bij([V]) = [\loopmodule[\Gamma]{V}{H}] \).
			Then
			\[
				\loopmodule[\Gamma]{\widetilde{V}}{H} \iso \loopmodule[\Gamma]{(V^{+(\gamma+H)})^f}{H}
				\iso ((\loopmodule[\Gamma]{V}{H})^{+\gamma})^f
			\]
			using \Cref{prop:twistproperties}~\ref{prop:VH-=V-H}
			and \Cref{prop:Hparityshiftimmediateobservations}~\ref{prop:Linoutshift}.
			That is, \( [\loopmodule[\Gamma]{\widetilde{V}}{H}]=[\loopmodule[\Gamma]{V}{H}] \)
			and \( \bij \) is well-defined.
	\end{itemize}

	For injectivity, assume now that \( \bij([\widetilde{V}]) = \bij([V]) \).
	We know that \( \widetilde{V} \) is \( \Gamma/H \)-graded irreducible
	but that \( \loopmodule[\Gamma]{V}{H} \) is not (by~\Cref{thm:VHdirectsum}),
	so \( \widetilde{V} \) cannot be equivalent to \( \loopmodule[\Gamma]{V}{H} \) in \( \irrep{\Gamma,N} \)
	(twists and parity shifts both preserve irreducibility---\Cref{prop:twistproperties}~\ref{prop:V-irred} and \Cref{prop:Hparityshiftimmediateobservations}~\ref{prop:W+hirred}).
	Thus, there are two cases:
	\begin{itemize}
		\item Assume \( \bij(\widetilde{V}) = [\widetilde{V}_\Gamma] =  \bij([V]) = [V_\Gamma]  \).
			We have that \( \widetilde{V}_\Gamma\iso (V_\Gamma^{+\gamma})^f \)
			for some \( \gamma\in\Gamma \) and some \( f\in\pdual{G} \).
			Coarsening to a \( \Gamma/H \)-grading,
			we find that \( \widetilde{V}\iso (V^{+(\gamma+H)})^f \).
			That is, \( [\widetilde{V}] = [V] \)
			and \( \bij \) is injective in this case.
		\item Assume \( \bij(\widetilde{V}) = [\loopmodule[\Gamma]{\widetilde{V}}{H}] = \bij([V]) = [\loopmodule[\Gamma]{V}{H}] \).
			Then
			\[
				\loopmodule[\Gamma]{\widetilde{V}}{H}\iso ((\loopmodule[\Gamma]{V}{H})^{+\gamma})^f
				\iso \loopmodule[\Gamma]{(V^{+(\gamma+H)})^f}{H}
			\]
			for some \( \gamma\in\Gamma \) and some \( f\in\pdual{G} \)
			(and using \Cref{prop:twistproperties}~\ref{prop:VH-=V-H}
			and \Cref{prop:Hparityshiftimmediateobservations}~\ref{prop:Linoutshift}).
			By \Cref{thm:VHdirectsum},
			we know that \( (V^{+(\gamma+H)})^f \)
			is a \( \Gamma/H \)-graded submodule of \( \loopmodule[\Gamma]{(V^{+(\gamma+H)})^f}{H} \)
			and also that
			\[
				\loopmodule[\Gamma]{\widetilde{V}}{H}\iso \bigoplus_{f_0\orthog{H}\in\pdualtwist}\widetilde{V}^{f_0}.
			\]
			Therefore, \( (V^{+(\gamma+H)})^f\iso \widetilde{V}^{f_0} \)
			for some \( f_0\in\pdual{\Gamma} \).
			Let \( \nu\colon G\to \Gamma \) be the natural map;
			we find that \( f_0\circ \nu\in \pdual{G} \).
			Therefore, \( (V^{+(\gamma+H)})^{f(f_0\circ\nu)^{-1}}\iso \widetilde{V} \)
			so \( [V] = [\widetilde{V}] \) and \( \bij \) is injective.
	\end{itemize}
	For surjectivity, take \( [W]\in\irrep{G,N} \).
	Using \Cref{thm:coarseirrep} and \Cref{thm:VHsemisimpleGmodule},
	we have that \( W \) appears as a \( \Gamma \)-graded submodule of \( \loopmodule[\Gamma]{V}{H} \) for some irreducible \( \Gamma/H \)-graded module \( V \).
	By \Cref{cor:VorVHirred},
	the only \( \Gamma \)-graded submodules of \( \loopmodule[\Gamma]{V}{H} \) are either \( \loopmodule[\Gamma]{V}{H} \) itself
	or \( V_\Gamma \).
	If \( W\iso\loopmodule[\Gamma]{V}{H} \)
	then the loop module is irreducible and so \( V \) is not \( \Gamma \)-gradable
	(by \Cref{prop:gradesubrep}).
	Thus, \( [W] = B([V]) \) in this case.
	Similarly, if  \( W\iso V_\Gamma \) then \( [W] = B([V]) \).
\end{proof}

Recall that a similar bijection was noted in~\cite[Remark~7.2]{EK2017} via the use of Brauer invariants.
We did not make use of Bauer invariants,
but instead assumed that \( H \) is a cyclic group of prime order.
This assumption gave the bijection \( \bij \) the very simple form~\eqref{eq:bijection}.

Taking \( \Gamma=\Ztzt,\, H = \{00,11\} \), the above theorem (in combination with discolouration\slash{}recolouration) gives a way to construct the graded irreducible modules of the \( \Ztzt \)-graded colour Lie superalgebras from the \( \Ztwo \)-graded Lie superalgebras.
More generally, we can iterate the theorem to produce a bijection \( \bij_1\colon\irrep{G,G} \to \irrep{G,\{0\}} \).
The image of this bijection contains
all of the \( G \)-graded irreducible modules (up to equivalence).

\begin{rmk}
	For any \( G \)-graded representation, all twists by a character in \( \pdual{G} \) are \( G \)-graded isomorphic; 
	therefore, we only need to consider the \( G \)-parity shifts (and not twists by a character) to determine the equivalence classes of \( \irrep{G,\{0\}} \).
	Similarly, there are no non-trivial \( G/G \)-parity shifts; hence, we only need to consider twists by elements of \( \pdual{G} \) to determine the equivalence classes of \( \irrep{G,G} \) in this case.
\end{rmk}

\begin{rmk}
	Using the notation from \Cref{thm:bijection},
	suppose that \( V\iso \loopmodule[\Gamma/H]{U}{\Gamma/H} \)
	for some ungraded \( \R \)-module \( U \).
	Notice that \( \loopmodule[\Gamma]{\loopmodule[\Gamma/H]{U}{\Gamma/H}}{H} \iso \loopmodule[\Gamma]{U}{\Gamma} \). Indeed,
	\[
		u\otimes e_{\gamma+H} \otimes e_\gamma \mapsto u\otimes e_\gamma
	\]
	is an isomorphism
	for \( u\in U \) and \( \gamma\in \Gamma \).
	If \( V \) is not \( \Gamma \)-gradable
	then
	\[
		B(V) = [\loopmodule[\Gamma]{V}{H}] = [\loopmodule[\Gamma]{U}{\Gamma}]
		= [\loopmodule[G/N]{U}{G/N}].
	\]
	Thus, if we iterate \( B \) applied to an ungraded module \( U \) and obtain a loop module at every step of the iteration,
	we may equivalently take a single (non-iterated) loop module with respect to a larger group.
\end{rmk}

In general, it may be difficult to determine which case of \( \bij \) in~\eqref{eq:bijection} applies.
There are multiple options we can test:
\begin{itemize}
	\item Try to find a \( \Gamma \)-grading for \( V \) (see~\cite{EK2015} for some results about gradability).
	\item Try to find \( \Gamma \)-graded proper submodules of \( \loopmodule[\Gamma]{V}{H} \) (all of which will provide a \( \Gamma \)-grading for \( V \)) or else prove that \( \loopmodule[\Gamma]{V}{H} \) is \( \Gamma \)-graded irreducible (and hence \( V \) will not be \( \Gamma \)-gradable).
	\item Check whether the twists of \( V \) by elements of \( \pdual{\Gamma} \) are isomorphic. This is equivalent to finding a \( \Gamma \)-grading by \Cref{prop:V-=Vifgrade}.
\end{itemize}
This third option (\Cref{prop:V-=Vifgrade}) is particularly useful if we have a classification of the \( \Gamma/H \)-graded irreducible modules (which we can obtain from a classification of the ungraded irreducible modules via iteration).
If we are instead iterating \( \bij^{-1} \) to find ungraded representations from graded ones,
then \Cref{prop:V+h=Vifred} plays a role analogous to \Cref{prop:V-=Vifgrade} in determining which case in~\eqref{eq:bijection} applies.
We will see the use of \Cref{prop:V-=Vifgrade,prop:V+h=Vifred} for this purpose in the example of~\Cref{sec:egD}.

\section{Example: colour \texorpdfstring{\( \spl[2] \)}{sl2}}\label{sec:exampleslc2}
As an example of applying \Cref{thm:bijection}, we will determine the finite-dimensional irreducible representations of \( \splc[2] \) over \( \CC \).
The algebra \( \splc[2] \) is a \( \Ztzt \)-graded Lie colour algebra with commutation factor \( \varepsilon(\alpha_1\alpha_2,\beta_1\beta_2) = (-1)^{\alpha_1\beta_2 - \alpha_2\beta_1} \) and sectors spanned by the following elements:
\begin{align*}
	00\text{-sector:}&~0  &
	10\text{-sector:}&~a_1  &
	01\text{-sector:}&~a_2  &
	11\text{-sector:}&~a_3
\end{align*}
and equipped with a bracket \( \cbrak{\cdot}{\cdot} \) which satisfies
\begin{align*}
	\cbrak{a_1}{a_2} &= a_3 = \cbrak{a_2}{a_1} & 
	\cbrak{a_2}{a_3} &= a_1 = \cbrak{a_3}{a_2} & 
	\cbrak{a_3}{a_1} &= a_2 = \cbrak{a_1}{a_3}.
\end{align*}
Using the multiplier \( \sigma(\alpha_1\alpha_2,\beta_1\beta_2) = (-1)^{\alpha_2\beta_1} \),
\( \splc[2] \) discolours to a Lie algebra with Lie bracket \( \sbrak{\cdot}{\cdot} = \comm{\cdot}{\cdot} \) given by
\begin{align*}
	\comm{a_1}{a_2} &= a_3 & 
	\comm{a_2}{a_3} &= -a_1 & 
	\comm{a_3}{a_1} &= -a_2.
\end{align*}
These commutation relations are satisfied by the following realisation
\begin{align}
	a_1 &= \frac{i}{2}(e - f) & a_2 &= -\frac{1}{2}(e+f) & a_3 &= -\frac{i}{2}h. \label{eq:splc2ahef}
\end{align}
in terms of a standard basis \( \{h,e,f\} \) for \( \spl[2] \), (\(\comm{h}{e} = 2e,\, \comm{h}{f} = -2f,\, \comm{e}{f} = h\)). Consequently, the discolouration \( (\splc[2])^{\sigma} \) is isomorphic to \( \spl[2] \) (hence the notation).

Both the ungraded~\cite{CSVO2006} and graded~\cite{BS2022} irreducible representations have appeared in the literature.
In~\cite{CSVO2006}, the authors found all the ungraded irreducible representations for \( \splc[2] \)
by embedding the universal enveloping algebra \( \Uea(\splc[2]) \) in \( \Mat[2](\Uea(\spl[2])) \) (the space of \( 2\times 2 \) matrices with entries in \( \Uea(\spl[2]) \)).
In~\cite{BS2022}, the corresponding graded irreducible representations for the \( \Ztzt \)-graded (non-colour) \( \spl[2] \) were obtained using ideals generated by central elements of the universal enveloping algebra \( \Uea(\spl[2]) \).
The approach that we will take is presented in the more general context of the loop module construction.

Our strategy is as follows:
\begin{enumerate}
	\item find all of the ungraded irreducible representations for \( \spl[2] \);
	\item use the techniques from the preceding sections to find the \( \Ztwo \)-graded irreducible representations for \( \spl[2] \); 
	\item similarly, find the \( \Ztzt \)-graded irreducible representations for \( \spl[2] \);
	\item recolour to obtain the \( \Ztzt \)-graded irreducible representations for \( \splc[2] \);
	\item find the ungraded irreducible representations for \( \splc[2] \) as subrepresentations.
\end{enumerate}

\subsection{Irreducible representations for \texorpdfstring{\( \spl[2] \)}{sl2}}
Recall that \( \spl[2] \) has a unique irreducible representation \( V_{\lambda} \) of dimension \( \lambda + 1 \) for every highest weight \( \lambda\in\ZZ_{\geq 0} \). 
We will choose a basis \( \{v_0,v_1,\ldots,v_{\lambda}\} \) for \( V_{\lambda} \) such that
\begin{align*}
	h v_j &= (\lambda - 2j) v_j & 
	e v_j &= (\lambda-j+1) v_{j-1} &
	f v_j &= (j+1) v_{j+1}
\end{align*}
with the convention \( v_{-1} = v_{\lambda+1} = 0 \).
In particular, note that \( v_0 \) is the highest weight vector.
Using~\eqref{eq:splc2ahef}, we find that
\begin{align*}
	a_1 v_j &= \frac{i}{2}((\lambda-j+1)v_{j-1} - (j+1)v_{j+1}) \\
	a_2 v_j &= -\frac{1}{2}((\lambda-j+1)v_{j-1} + (j+1)v_{j+1}) \\
	a_3 v_j &= -\frac{i}{2}(\lambda-2j) v_j.
\end{align*}

\subsection{\texorpdfstring{\( \Ztwo \)}{Z2}-graded irreducible representations}
We wish to find the \( \Ztwo \)-graded representations.
So, we choose \( \Gamma_1 = H_1 = \Ztwo \times \Ztwo/\{00,11\} \iso \Ztwo \).
Under this group and the original \( \splc[2] \) grading, the sectors for \( \spl[2] \) are spanned by the following elements:
\begin{align*}
	0\text{-sector:}&~a_3 && 1\text{-sector:}~a_1,a_2.
\end{align*}
The ungraded module \( V_{\lambda} \) can be given a \( \Gamma_1 \)-grading: \( V_{\lambda,0}^\mathrm{E}\oplus V_{\lambda,1}^\mathrm{E} \) where
\begin{align*}
	V_{\lambda,0}^\mathrm{E} &= \spn\{v_j\mid j~\text{even}\} &
	V_{\lambda,1}^\mathrm{E} &= \spn\{v_j\mid j~\text{odd}\}.
\end{align*}
We can easily check that \( a_1 V_{\lambda,\gamma}^\mathrm{E} \subseteq V_{\lambda,\gamma+1}^\mathrm{E},\, a_2V_{\lambda,\gamma}^\mathrm{E} \subseteq V_{\lambda,\gamma+1}^\mathrm{E} \) and \( a_3 V_{\lambda,\gamma}^\mathrm{E} \subseteq V_{\lambda,\gamma}^\mathrm{E} \) for \( \gamma \in \Ztwo \).
Call this \( \Gamma_1 \)-graded module \( V_{\lambda}^\mathrm{E} \).

Since \( V_{\lambda} \) is ungraded irreducible,
\( V_{\lambda}^E \) is \( \Gamma_1 \)-graded irreducible.
By shifting the \( H_1 \)-parity of \( V_{\lambda}^\mathrm{E} \),
we get another \( \Gamma_1 \)-graded irreducible module \( V_{\lambda}^\mathrm{O} = V_{\lambda,0}^\mathrm{O} \oplus V_{\lambda,1}^\mathrm{O} \) where
\begin{align*}
	V_{\lambda,0}^\mathrm{O} &= \spn\{v_j\mid j~\text{odd}\} &
	V_{\lambda,1}^\mathrm{O} &= \spn\{v_j\mid j~\text{even}\}.
\end{align*}
By \Cref{prop:gradinguniqueness}, we do not need to search for any other gradings.
And by \Cref{prop:gradesubrep},
we do not need to check the loop module at all.
Thus, the finite-dimensional \( \Gamma_1 \)-graded irreducible \( \spl[2] \)-modules are \( V_{\lambda}^\mathrm{E} \) and \( V_{\lambda}^\mathrm{O} \) for \( \lambda\in\ZZ_{\geq 0} \).

\subsection{\texorpdfstring{\( \Ztzt \)}{Z2xZ2}-graded irreducible representations}
Now we wish to find the \( \Ztzt \)-graded representations.
So, we choose \( \Gamma_2 = \Ztzt \) and \( H_2 = \{00,11\} \).
We will now find the \( \Gamma_2 \)-gradings, with the aid of the following lemma.

\begin{lem}\label{lem:v0vl}
	Let \( V_{\lambda} = \bigoplus_{\gamma\in\Gamma_2}W_{\gamma} \) be a \( \Gamma_2 \)-grading for \( V_{\lambda} \).
	Then \( v_0 = \sum_{\gamma\in\Gamma_2} v_{0,\gamma} \) where \( v_{0,\gamma}\in W_{\gamma} \) is some linear combination of \( v_0 \) and \( v_{\lambda} \) (possibly equal to zero).
\end{lem}
\begin{proof}
	Since \( h = 2ia_3 \) is homogeneous of degree \( 11 \),
	we find that \( h^2 W_{\gamma} \subseteq W_{\gamma} \).
	If we project \( v_0 \) onto each sector and write \( v_0 = \sum_{\gamma\in\Gamma_2} v_{0,\gamma} \),
	for \( v_{0,\gamma}\in W_{\gamma} \) then we find that
	\[
		\sum_{\gamma\in\Gamma_2}h^2v_{0,\gamma} = h^2(v_0) = \sum_{\gamma\in\Gamma_2} \lambda^2 v_{0,\gamma}.
	\]
	In particular, \( h^2v_{0,\gamma} = \lambda^2 v_{0,\gamma} \) for every \( \gamma\in\Gamma_2 \).
	Using our knowledge of the eigenvectors of \( h \),
	we find that \( v_{0,\gamma} \) must be a linear combination of \( v_0 \) and \( v_{\lambda} \).
\end{proof}
\begin{rmk}
	A similar result holds for \( v_{\lambda} = \sum_{\gamma\in\Gamma_2}v_{\lambda,\gamma} \).
\end{rmk}
\Cref{lem:v0vl} tells us that at least one sector of any \( \Gamma_2 \)-grading must contain a nonzero linear combination of \( v_0 \) and \( v_{\lambda} \).
Our strategy is to then apply \( a_1 \) and \( a_2 \) to this linear combination to help deduce the structure of a \( \Gamma_2 \)-grading or prove that one cannot exist.

\subsubsection{The \texorpdfstring{\( \lambda \)}{lambda} even case}
If \( \lambda \) is even, then both \( v_0 \) and \( v_{\lambda} \) are in \( V_{\lambda,0}^{\mathrm{E}} \).
Therefore, using \Cref{lem:v0vl}, we find that the \( 00 \)-sector or \( 11 \)-sector must contain a nonzero linear combination of \( v_0 \) and \( v_{\lambda} \).
We choose (after some trial and error) \( v_0 + v_{\lambda} \) to be in the \( 00 \)-sector and \( v_0 - v_{\lambda} \) to be in the \( 11 \)-sector.
With this choice, we can find a \( \Gamma_2 \)-grading for \( V_{\lambda}^\mathrm{E} \):
\begin{equation}\label{eq:VlE+sectors}
	\begin{aligned}
		V_{\lambda,00}^{\mathrm{E}+} &= \spn\{v_j + v_{\lambda-j}\mid j~\text{even}\}, &
		V_{\lambda,10}^{\mathrm{E}+} &= \spn\{v_j - v_{\lambda-j}\mid j~\text{odd}\}, \\
		V_{\lambda,01}^{\mathrm{E}+} &= \spn\{v_j + v_{\lambda-j}\mid j~\text{odd}\}, &
		V_{\lambda,11}^{\mathrm{E}+} &= \spn\{v_j - v_{\lambda-j}\mid j~\text{even}\}.
	\end{aligned}
\end{equation}
Denote this \( \Gamma_2 \)-graded module by \( V_{\lambda}^{\mathrm{E}+} \).
We calculate
\begin{align*}
	a_1(v_j \pm v_{\lambda-j}) &= \frac{i}{2}((\lambda-j+1)(v_{j-1} \mp v_{\lambda-(j-1)}) - (j+1)(v_{j+1}\mp v_{\lambda-(j+1)}) \\
	a_2(v_j \pm v_{\lambda-j}) &= -\frac{1}{2}((\lambda-j+1)(v_{j-1} \pm v_{\lambda-(j-1)}) + (j+1)(v_{j+1}\pm v_{\lambda-(j+1)}) \\
	a_3(v_j \pm v_{\lambda-j}) &= -\frac{i}{2}((\lambda-2j)(v_j \mp v_{\lambda-j}))
\end{align*}
from which we can easily verify that \( a_1 V_{\lambda,\gamma}^{\mathrm{E}+} \subseteq V_{\lambda,\gamma + 10}^{\mathrm{E}+} \), \( a_2 V_{\lambda,\gamma}^{\mathrm{E}+}\subseteq V_{\lambda,\gamma + 01}^{\mathrm{E}+} \) and \( a_3 V_{\lambda,\gamma}^{\mathrm{E}+} \subseteq V_{\lambda,\gamma + 11}^{\mathrm{E}+} \) for each \( \gamma\in\Gamma_2 \).
Since \( V_{\lambda}^\mathrm{E} \) is \( \Gamma_2 \)-graded irreducible, so is \( V_{\lambda}^{\mathrm{E}+} \).
By shifting the \( H_2 \)-parity of \( V_{\lambda}^{\mathrm{E}+} \),
we get another \( \Gamma_2 \)-graded irreducible module \( V_{\lambda}^{\mathrm{E}-} \),
which can be obtained from \( V_{\lambda}^{\mathrm{E}+} \) by swapping the \( 00 \)- with the \( 11 \)-sector and the \( 10 \)- with the \( 01 \)-sector.
Note that \( V_{\lambda}^{\mathrm{E}+} \) and \( V_{\lambda}^{\mathrm{E}-} \) are not isomorphic,
because \( v_{\lambda/2} = (1/2)(v_{\lambda/2} + v_{\lambda - \lambda/2}) \)  (the unique vector with weight \( 0 \)) will be in different sectors.

Performing similar computations with \( V_{\lambda}^\mathrm{O} \),
we find two more \( \Gamma_2 \)-graded irreducible modules \( V_{\lambda}^{\mathrm{O}+} \) and \( V_{\lambda}^{\mathrm{O}-} \) (these can be obtained from \( V_{\lambda}^{\mathrm{E}+} \) by swapping the \( 00 \)- with the \( 01 \)-sector and the \( 10 \)- with the \( 11 \)-sector to get \( V_{\lambda}^{\mathrm{O}+} \);
and by swapping the \( 00 \)- with the \( 10 \)-sector and the \( 01 \)- with the \( 11 \)-sector to get \( V_{\lambda}^{\mathrm{O}-} \)).
By \Cref{prop:gradinguniqueness} and \Cref{thm:bijection} we do not need to search for any more \( \Gamma_2 \)-graded irreducible modules in the case where \( \lambda \) is even.

\subsubsection{The \texorpdfstring{\( \lambda \)}{lambda} odd case}
\begin{lem}
	If \( \lambda \) is odd, then \( V_{\lambda}^{\mathrm{E}} \) cannot be given a \( \Gamma_2 \)-grading.
\end{lem}
\begin{proof}
	For contradiction, assume that we can give \( V_{\lambda}^{\mathrm{E}} \) a \( \Gamma_2 \)-grading, \( V_{\lambda}^{\mathrm{E}} = \bigoplus_{\gamma\in\Gamma_2} V_{\lambda,\gamma}^{\mathrm{E}} \).
	From the structure of \( V_{\lambda}^{\mathrm{E}} \), we know that \( v_0 \in V_{\lambda,0}^{\mathrm{E}} = V_{\lambda,00}^{\mathrm{E}} \oplus V_{\lambda,11}^{\mathrm{E}} \).
	Therefore, we can write \( v_0 = v_{0,00} + v_{0,11} \) for \( v_{0,00}\in V_{\lambda,00}^{\mathrm{E}},\, v_{0,11}\in V_{\lambda,11}^{\mathrm{E}} \).
	By \Cref{lem:v0vl}, both \( v_{0,00} \) and \( v_{0,11} \) are linear combinations of \( v_0 \) and \( v_{\lambda} \).
	But since \( v_{\lambda} \in V_{\lambda,1}^{\mathrm{E}} = V_{\lambda,10}^{\mathrm{E}} \oplus V_{\lambda,01}^{\mathrm{E}} \), we must have that both \( v_{0,00} \) and \( v_{0,11} \) are scalar multiples of \( v_0 \).
	And since they sum to \( v_0 \neq 0 \), at least one of them must be nonzero.
	That is, \( v_0\in V_{\lambda,00}^{\mathrm{E}} \) or \( v_0 \in V_{\lambda,11}^{\mathrm{E}} \).
	An analogous argument shows that \( v_{\lambda}\in V_{\lambda,10} \) or \( v_{\lambda}\in V_{\lambda,01} \).

	If \( v_0 \in V_{\lambda,00}^{\mathrm{E}} \) and \( v_{\lambda} \in V_{\lambda,01}^{\mathrm{E}} \) 
	then applying \( a_1 \) we find
	\begin{align*}
		v_1 &= 2ia_1 v_0 \in V_{\lambda,00}^{\mathrm{E}} \oplus V_{\lambda,10}^{\mathrm{E}} \\
		v_2 &= ia_1 v_1 + \frac{\lambda}{2} v_0   \in V_{\lambda,00}^{\mathrm{E}} \oplus V_{\lambda,10}^{\mathrm{E}} \\
		    &\ \vdots \\
		v_{j+1} &= \frac{2i}{j+1} a_1 v_j + \frac{\lambda-j+1}{j+1} v_{j-1} \in V_{\lambda,00}^{\mathrm{E}} \oplus V_{\lambda,10}^{\mathrm{E}}.
	\end{align*}
	Proceeding inductively, we find that \( V_{\lambda}^{\mathrm{E}}\subseteq V_{\lambda,00}^{\mathrm{E}} \oplus V_{\lambda,10}^{\mathrm{E}} \).
	But, by applying \( a_1 \) to \( v_{\lambda} \) in a similar way, we find that \( V_{\lambda}^{\mathrm{E}}\subseteq V_{\lambda,01}^{\mathrm{E}} \oplus V_{\lambda,11}^{\mathrm{E}} \), a contradiction!
	We similarly arrive at a contradiction if we assume \( v_0\in V_{\lambda,11}^{\mathrm{E}} \) and \( v_{\lambda} \in V_{\lambda,10}^{\mathrm{E}} \); \( v_0 \in V_{\lambda,00}^{\mathrm{E}} \) and \( v_{\lambda} \in V_{\lambda,10}^{\mathrm{E}} \); or \( v_0\in V_{\lambda,11} \) and \( v_{\lambda} \in V_{\lambda,01} \).
	(Note that in some of these cases we need to use \( a_2 \) instead of \( a_1 \).)
\end{proof}

We have just shown that \( V_{\lambda}^{\mathrm{E}} \) cannot be given a \( \Gamma_2 \)-grading if \( \lambda \) is odd.
By \Cref{cor:VorVHirred}, the loop module \( \loopmodule[\Ztzt]{V_{\lambda}^{\mathrm{E}}}{H_2} \) must therefore be a \( \Gamma_2 \)-graded irreducible representation.
If we take a basis \( \{v_{\alpha,j} \mid \alpha\in H_2\iso\Ztwo,\, j =0,\ldots,\lambda\} \) for \( \loopmodule[\Gamma_2]{V_{\lambda}^{\mathrm{E}}}{H_2} \), where
\begin{equation}\label{eq:VEH2sectors}
	\begin{aligned}
		\loopmodule[\Gamma_2]{V^{\mathrm{E}}_\lambda}{H_2}_{00} &= \spn\{v_{0,j}\mid j~\text{even}\}, &
		\loopmodule[\Gamma_2]{V^{\mathrm{E}}_\lambda}{H_2}_{10} &= \spn\{v_{1,j}\mid j~\text{odd}\}, \\
		\loopmodule[\Gamma_2]{V^{\mathrm{E}}_\lambda}{H_2}_{01} &= \spn\{v_{0,j}\mid j~\text{odd}\}, &
		\loopmodule[\Gamma_2]{V^{\mathrm{E}}_\lambda}{H_2}_{11} &= \spn\{v_{1,j}\mid j~\text{even}\},
	\end{aligned}
\end{equation}
then we can compute the action as
\begin{align*}
	a_1 v_{\alpha,j} &= \frac{i}{2}((\lambda-j+1)v_{\alpha+1,j-1} - (j+1)v_{\alpha+1,j+1}) \\
	a_2 v_{\alpha,j} &= -\frac{1}{2}((\lambda-j+1)v_{\alpha,j-1} + (j+1)v_{\alpha,j+1}) \\
	a_3 v_{\alpha,j} &= -\frac{i}{2}(\lambda - 2j)v_{\alpha+1,j}.
\end{align*}

By similar reasoning, we find that \( \loopmodule[\Gamma_2]{V_{\lambda}^{\mathrm{O}}}{H_2} \) is a \( \Gamma_2 \)-graded irreducible representation.
We can obtain \( \loopmodule[\Gamma_2]{V_{\lambda}^{\mathrm{O}}}{H_2} \) from \( \loopmodule[\Gamma_2]{V_{\lambda}^{\mathrm{E}}}{H_2} \) by swapping the \( 00 \)- with the \( 01 \)-sector and the \( 10 \)- with the \( 11 \)-sector.
Note that if we shift the \( H_2 \)-parity of either of these modules, we get the same module.
By \Cref{thm:bijection}, the only \( \Gamma_2 \)-graded irreducible modules for the case where \( \lambda \) is odd are \( \loopmodule[\Gamma_2]{V_{\lambda}^{\mathrm{E}}}{H_2} \) and \( \loopmodule[\Gamma_2]{V_{\lambda}^{\mathrm{O}}}{H_2} \).

\subsection{Recolouring}
To recolour \( \spl[2] \) to \( \splc[2] \) we use the multiplier \( \sigma(\alpha_1\alpha_2,\beta_1\beta_2) = (-1)^{\alpha_2\beta_1} \).
To obtain the \( \Ztzt \)-graded irreducible modules for \( \splc[2] \), we use this same multiplier for the \( \Ztzt \)-graded irreducible modules for \( \spl[2] \).

\begin{thm}\label{thm:slc2Ztztirreps}
	The only finite-dimensional \( \Ztzt \)-graded irreducible modules for \( \splc[2] \) are:
	\begin{enumerate}[label=(\roman*)]
		\item 
			For each even number \( \lambda \),
			the four modules which are equivalent to \( (V_{\lambda}^{\mathrm{E}+})^{\sigma} \) up to \( \Ztzt \)-parity shift.
			The sectors of \( (V_{\lambda}^{\mathrm{E}+})^{\sigma} \) are given as in \eqref{eq:VlE+sectors} and the action is given by
			\begin{align*}
				a_1(v_j \pm v_{\lambda-j}) &= \frac{i}{2}((\lambda-j+1)(v_{j-1} \mp v_{\lambda-(j-1)}) - (j+1)(v_{j+1}\mp v_{\lambda-(j+1)}) \\
				a_2(v_j \pm v_{\lambda-j}) &= \mp\frac{1}{2}((\lambda-j+1)(v_{j-1} \pm v_{\lambda-(j-1)}) + (j+1)(v_{j+1}\pm v_{\lambda-(j+1)}) \\
				a_3(v_j \pm v_{\lambda-j}) &= \mp\frac{i}{2}((\lambda-2j)(v_j \mp v_{\lambda-j})).
			\end{align*}
			The dimension of these modules is \( \lambda+1 \).
		\item 
			For each odd number \( \lambda \),
			the two modules which are equivalent to \( (\loopmodule[\Gamma_2]{V_{\lambda}^{\mathrm{E}}}{H_2})^{\sigma} \) up to \( \{00,01\} \)-parity shift.
			The sectors of \( (\loopmodule[\Gamma_2]{V_{\lambda}^{\mathrm{E}}}{H_2})^{\sigma} \) are given as in \eqref{eq:VEH2sectors} and the action is given by
			\begin{align*}
				a_1 v_{\alpha,j} &= \frac{i}{2}((\lambda-j+1)v_{\alpha+1,j-1} - (j+1)v_{\alpha+1,j+1}) \\
				a_2 v_{\alpha,j} &= -\frac{(-1)^{\alpha}}{2}((\lambda-j+1)v_{\alpha,j-1} + (j+1)v_{\alpha,j+1}) \\
				a_3 v_{\alpha,j} &= -\frac{i(-1)^{\alpha}}{2}(\lambda - 2j)v_{\alpha+1,j}.
			\end{align*}
			The dimension of these modules is \( 2(\lambda+1) \).
	\end{enumerate}
\end{thm}
\begin{proof}
	Use the results from the previous sections and the facts that \( \sigma(10,\beta) = 1 \) for all \( \beta\in\Ztzt \) and 
	\[
		\sigma(01,\beta) = \sigma(11,\beta) =
		\begin{cases}
			-1 & \text{if}~\beta=10,11\\
			1 & \text{if}~\beta=00,01.
		\end{cases}
		\qedhere
	\]
\end{proof}

The above theorem proves, in addition, that these modules are the \emph{only} \( \Ztzt \)-graded irreducible \( \splc[2] \) modules.

Note that, for each \( \lambda\in\ZZ_{\geq0} \),
there is a unique equivalence class of \( \Ztzt \)-graded irreducible \( \splc[2] \)-modules.
This corresponds exactly to the unique irreducible \( \spl[2] \)-modules,
demonstrating the bijection in \Cref{thm:bijection}.

\subsection{Ungraded irreducible representations for \texorpdfstring{\( \splc[2] \)}{slc2}}
\begin{lem}
	The ungraded irreducible representations for \( \splc[2] \)
	all appear as ungraded subrepresentations of the \( \Ztzt \)-graded irreducible representations in \Cref{thm:slc2Ztztirreps}.
\end{lem}
\begin{proof}
	If we knew the ungraded irreducible representations for \( \splc[2] \),
	then we could use the same process in the preceding sections to derive the same \( \Ztzt \)-graded irreducible modules.
	By applying \Cref{thm:bijection} twice, these \( \Ztzt \)-graded irreducible modules are either irreducible as ungraded modules, a loop module or a loop module of a loop module.
	In any case, the original ungraded irreducible modules must appear as ungraded submodules of the \( \Ztzt \)-graded modules.
\end{proof}

Note that \( \Ztzt \)-graded modules which are equivalent up to \( H \)-parity shift (for some \( H\leq \Ztwo \)) are isomorphic as ungraded modules.

\begin{lem}\label{lem:VlE+irrep}
	If \( \lambda \) is even, then \( (V_{\lambda}^{E+})^{\sigma} \) is irreducible as an ungraded module.
\end{lem}
\begin{proof}
	Choose a new basis \( \{u_j\mid j=0,\ldots,\lambda\} \) 
	where \( u_j = (v_j + v_{\lambda-j}) + i(-1)^j(v_j-v_{\lambda-j}) \).
	Note that this is not linearly dependent since \( u_{\lambda-j} = (v_j + v_{\lambda-j}) - i(-1)^j(v_j-v_{\lambda-j}) \).
	The action of \( \splc[2] \) is then
	\begin{equation}\label{eq:ungradedVlE+}
		\begin{aligned}
			a_1u_j &= \frac{(-1)^{j+1}}{2}((\lambda-j+1) u_{j-1} - (j+1) u_{j+1})\\
			a_2u_j &= -\frac{1}{2}((\lambda-j+1) u_{j-1} + (j+1) u_{j+1})\\
			a_3u_j &= \frac{(-1)^{j+1}}{2}(\lambda-2j) u_j
		\end{aligned}
	\end{equation}
	with the convention that \( u_{-1} = u_{\lambda+1} = 0 \).
	In particular, 
	\begin{align*}
		(a_2 + a_1)u_j &= 
		\begin{cases}
			-(\lambda-j+1) u_{j-1} & \text{if \( j \) even} \\
			-(j+1) u_{j+1} & \text{if \( j \) odd}.
		\end{cases} \\
		(a_2 - a_1)u_j &= 
		\begin{cases}
			-(j+1) u_{j+1} & \text{if \( j \) even}\\
			-(\lambda-j+1) u_{j-1} & \text{if \( j \) odd}. \\
		\end{cases}
	\end{align*}

	So, if we take an arbitrary element \( \sum_{j} c_j u_j \) of \( (V_{\lambda}^{E+})^{\sigma} \),
	then we can alternate applying \( (a_2+a_1) \) and \( (a_2-a_1) \) 
	(i.e.\ we apply \( ((a_2-a_1)(a_2+a_1))^k \) or \( (a_2-a_1)((a_2+a_1)(a_2-a_1))^k \) for some non-negative integer \( k \))
	to raise/lower the basis elements of \( \sum_j c_j u_j \).
	Eventually, this procedure will yield a scalar multiple of either \( u_0 \) or \( u_{\lambda} \),
	from which we can generate the entire module.
	Since the entire module can be generated from a single vector,
	\( (V_{\lambda}^{\mathrm{E}+})^{\sigma} \) is irreducible.
\end{proof}

Now we consider the case when \( \lambda \) is odd.
Then \( (\loopmodule[\Gamma_2]{V_{\lambda}^{\mathrm{E}}}{H_2})^{\sigma} = U_{\lambda}^{++} \oplus U_{\lambda}^{+-} \oplus U_{\lambda}^{-+} \oplus U_{\lambda}^{--} \) where
\begin{align*}
	U_{\lambda}^{\zeta\xi} &= \spn\{u^{\zeta\xi}_j \mid j=0,1,\ldots, (\lambda-1)/2\} \\
	u^{\zeta\xi}_j &= v_{0,j} + \zeta i (-1)^j v_{1,j} + \xi v_{0,\lambda-j} - \zeta\xi i (-1)^j v_{1,\lambda-j}
\end{align*}
for \( \zeta,\xi\in\{-1,+1\} \).
The action of \( \splc[2] \) is given by
\begin{equation}\label{eq:ungradedUzj}
	\begin{aligned}
		a_1 u^{\zeta\xi}_j & = - \frac{\zeta}{2}(-1)^j((\lambda-j+1) u^{\zeta\xi}_{j-1} - (j+1)u^{\zeta\xi}_{j+1}) && \left(j< \frac{\lambda-1}{2}\right) \\
		a_1 u^{\zeta\xi}_{(\lambda-1)/2} & = - \frac{\zeta}{2}(-1)^{(\lambda-1)/2}\left(\frac{\lambda+3}{2} u^{\zeta\xi}_{(\lambda-3)/2} - \xi\frac{\lambda+1}{2}u^{\zeta\xi}_{(\lambda-1)/2}\right) \\
		a_2 u^{\zeta\xi}_j &= -\frac{1}{2}((\lambda-j+1) u^{\zeta\xi}_{j-1} + (j+1)u^{\zeta\xi}_{j+1}) && \left(j<\frac{\lambda-1}{2}\right)\\
		a_2 u^{\zeta\xi}_{(\lambda-1)/2} &= -\frac{1}{2}\left(\frac{\lambda+3}{2} u^{\zeta\xi}_{(\lambda-3)/2} + \xi \frac{\lambda+1}{2}u^{\zeta\xi}_{(\lambda-1)/2}\right)\\
		a_3 u^{\zeta\xi}_j &= - \frac{\zeta}{2}(-1)^j(\lambda-2j) u^{\zeta\xi}_j
	\end{aligned}
\end{equation}
with the convention \( u_{-1} = 1 \).
(Note that \( u^{\zeta\xi}_j \) was chosen so that it was an eigenvector of \( a_3 \).)
Clearly \( U_{\lambda}^{\zeta\xi} \) gives rise to four non-isomorphic modules,
since the action of \( a_1 \), \( a_2 \) and \( a_3 \) on \( u^{\zeta\xi}_{(\lambda-1)/2} \)
(the unique eigenvector for \( a_3 \) which has eigenvalue with absolute value \( 1/2 \)) is different for different values of \( \zeta,\xi \).

\begin{lem}
	\( U_{\lambda}^{\zeta\xi} \) is irreducible as an ungraded module.
\end{lem}
\begin{proof}
	Observe that, for \( j< (\lambda-1)/2 \),
	\begin{align*}
		(a_2+a_1)u^{\zeta\xi}_j &= 
		\begin{cases}
			-(\lambda-j+1)u^{\zeta\xi}_{j-1} & \text{if}~\zeta(-1)^j = 1\\
			-(j+1)u^{\zeta\xi}_{j+1} & \text{if}~\zeta(-1)^j = -1
		\end{cases} \\
		(a_2+a_1)u^{\zeta\xi}_{(\lambda-1)/2} &= 
		\begin{cases}
			-\frac{\lambda+3}{2}u^{\zeta\xi}_{(\lambda-3)/2} & \text{if}~\zeta(-1)^{(\lambda-1)/2} = 1\\
			-\xi\frac{\lambda+1}{2}u^{\zeta\xi}_{(\lambda-1)/2} & \text{if}~\zeta(-1)^{(\lambda-1)/2} = -1
		\end{cases} \\
		(a_2-a_1)u^{\zeta\xi}_j &= 
		\begin{cases}
			-(j+1)u^{\zeta\xi}_{j+1} & \text{if}~\zeta(-1)^j = 1\\
			-(\lambda-j+1)u^{\zeta\xi}_{j-1} & \text{if}~\zeta(-1)^j = -1
		\end{cases} \\
		(a_2-a_1)u^{\zeta\xi}_{(\lambda-1)/2} &= 
		\begin{cases}
			-\xi\frac{\lambda+1}{2}u^{\zeta\xi}_{(\lambda-1)/2} & \text{if}~\zeta(-1)^{(\lambda-1)/2} = 1 \\
			-\frac{\lambda+3}{2}u^{\zeta\xi}_{(\lambda-3)/2} & \text{if}~\zeta(-1)^{(\lambda-1)/2} = -1
		\end{cases}
	\end{align*}
	Take an arbitrary element \( \sum_j c_j u_j \) of \( U_{\lambda}^{\zeta\xi} \).
	Similar to the proof of \Cref{lem:VlE+irrep} we alternate applying \( a_2+a_1 \) and \( a_2-a_1 \).
	In this case, we are able to obtain a scalar multiple of \( u_0 \) from \( \sum_j c_j u_j \),
	from which we can generate the entire module.
\end{proof}

In summary, we have proven the following theorem:
\begin{thm}\label{thm:slc2ungradedirreps}
	The only finite-dimensional ungraded irreducible modules for \( \splc[2] \) are 
	\begin{enumerate}[label=(\roman*)]
		\item For each even number \( \lambda \), the module \( (V_{\lambda}^{\mathrm{E}+})^{\sigma} \)
			with action given in \eqref{eq:ungradedVlE+}. The dimension of these modules is \( \lambda+1 \).
		\item For each odd number \( \lambda \) and each \( \zeta,\xi\in\{-1,+1\} \), the module \( U_{\lambda}^{\zeta\xi} \) with action given in \eqref{eq:ungradedUzj}. The dimension of these modules is \( (\lambda+1)/2 \).
	\end{enumerate}
\end{thm}

This theorem agrees with the previous result of~\cite[Theorem~5.2]{CSVO2006},
but is presented in the context of the theory in this paper.

\begin{rmk}
	It is clear that the modules \( U_{\lambda}^{+\xi} \) and \( U_{\lambda}^{-\xi} \) are twists of each other, with twisting character \( f\in\pdual{\Gamma} \) given by \( f(00)=f(01)=1 \) and \( f(10)=f(11)=-1 \) (here we can choose \( \Gamma = \Ztzt/\{00,01\} \) and \( H = \Gamma \)).
	If we instead choose 
	\( f(00)=f(11)=1 \) and \( f(01)=f(10)=-1 \)
	(with \( \Gamma = \Ztzt/\{00,11\},\, H = \Gamma \)),
	then we find that the twist of \( U_{\lambda}^{\zeta+} \) by \( f \) is isomorphic to \( U_{\lambda}^{\zeta-} \).
\end{rmk}

The above remark shows that \( U^{\zeta\xi} \) are all equivalent up to twist by \( \pdual{\Gamma} \) for different values of \( \zeta,\xi \).
Consequently, there is a unique equivalence class in \Cref{thm:slc2ungradedirreps} for each \( \lambda\in\ZZ_{\geq0} \).
This corresponds exactly to the equivalence classes for \( \Ztzt \)-graded \( \splc[2] \)-modules in \Cref{thm:slc2Ztztirreps},
again demonstrating the bijection in \Cref{thm:bijection}.

\begin{rmk}
	Although the ungraded \( \splc[2] \) representations were derived from the ungraded \( \spl[2] \) representations, they have a remarkably different structure. This shows that, despite the procedure outlined in this paper, colour algebras still have an interesting representation theory.
\end{rmk}

\section{Example: A class of \texorpdfstring{\( \Ztzt \)}{Z2xZ2}-graded Lie colour algebras of type D}\label{sec:egD}

\subsection{Definition of the algebra}
Let \( \sotwoc(m,n) \)
be the subspace of \( \Mat[(2m+2n)\times(2m+2n)] \) spanned by
matrices of the form
\begin{equation}\label{eq:sotwomatrices}
		\left[
			\begin{array}{cc|cc}
				A & -B & X & Y\\
				\trsp{B} & C & \trsp{Y} & Z\\
				\hline
				X & Y & A & -B\\
				\trsp{Y} & Z & \trsp{B} & C
			\end{array}
		\right]
\end{equation}
where \( A,X\in\Mat[m\times m],\, C,Z\in\Mat[n\times n] \) are anti-symmetric and \( B,Y\in\Mat[m\times n] \) are arbitrary.

The vector space \( \FF^m\times \FF^n \times \FF^m \times \FF^n  \)
can be given a \( \Ztzt \) grading by assigning gradings \( 00,10,01,11 \)
to the direct factors
(respectively).
We can then interpret~\eqref{eq:sotwomatrices}
as elements of \( \End(\FF^{2m+2n}) \);
the \( \Ztzt \)-grading of \( \End(\FF^{2m+2n}) \)
induces a \( \Ztzt \)-grading of \( \sotwoc(m,n) \).
Explicitly: the \( 00 \)-sector is spanned by matrices where the only nonzero sub-matrices are \( A \) and \( C \);
the \( 10 \)-sector has only \( B \) nonzero; the \( 01 \)-sector has only \( X \) and \( Z \) nonzero;
and the \( 11 \)-sector has only \( Y \) nonzero.

Via straightforward yet somewhat tedious computation, we can verify that \( \sotwoc(m,n) \) closes to form a Lie colour algebra with commutation factor \( \varepsilon(\alpha_1\alpha_2, \beta_1\beta_2) = (-1)^{\alpha_1\beta_2 - \alpha_2\beta_1} \)
(alternatively, see~\Cref{rmk:so2loopalgebra}).

For the example of this section, we will focus on the case that \( m+n \) is even.
Set \( r = (m+n)/2 \).
Note that the rank of \( \sotwoc(m,n) \) is \( 2r \) (this will become apparent shortly).

\subsection{Discolouring the algebra}

We will choose the multiplier \( \sigma(\alpha_1\alpha_2, \beta_1\beta_2) = (-1)^{\alpha_2\beta_1} \) to discolour \( \sotwoc(m,n) \).
We interpret the matrices~\eqref{eq:sotwomatrices} as a representation \( \sotwoc(m,n)\to\End(\FF^{2m+2n}) \)
and discolour this matrix representation.
The resulting discoloured matrix representation is spanned by
\[
		\left[
			\begin{array}{cc|cc}
				A & -B & X & -Y\\
				\trsp{B} & C & \trsp{Y} & -Z\\
				\hline
				X & -Y & A & -B\\
				\trsp{Y} & -Z & \trsp{B} & C
			\end{array}
		\right]
\]
and closes to form a Lie algebra with a \( \Ztzt \)-grading.
Denote this discoloured algebra by \( \sotwo(2r) \).

Set
\[
	Q = 
	\left[
		\begin{array}{c|c}
			I_{m+n} & I_{m+n}\\
			\hline
			I_{m+n} & - I_{m+n}
		\end{array}
	\right]
\]
where \( I_{m+n} \) is the \( (m+n)\times(m+n) \) identity matrix.
Using \( Q \) to change the basis, we find that \( Q\sotwo(2r)Q^{-1} \) is spanned by the matrices
\[
		\left[
			\begin{array}{cc|cc}
				A+X & -B-Y & 0 & 0\\
				-\trsp{(-B-Y)} & C-Z & 0 & 0\\
				\hline
				0 & 0 & A-X & -B+Y\\
				0 & 0 & -\trsp{(-B+Y)} & C+Z
			\end{array}
		\right]
\]
and therefore, if we ignore the grading, \( \sotwo(2r) \iso \so(2r)\oplus \so(2r) \).
(In fact, \( Q \) preserves the \( \Ztwo\iso\Ztzt/\{00,01\} \)-grading---but not the full \( \Ztzt \)-grading---%
so \( \sotwo(2r)\iso \so(2r)\oplus\so(2r) \) as \( \Ztwo \)-graded algebras.)

\begin{rmk}\label{rmk:so2loopalgebra}
	Even though \( \sotwo(2r) \) is not ungraded simple,
	it is \( \Ztzt \)-graded simple.
	Indeed, one can easily verify that \( \sotwo(2r) \) is the loop algebra of \( \Ztwo \)-graded \( \so(2r) \).
	The loop algebra of a graded-simple algebra is also graded simple \cite[Lemma~5.1.1]{ABFP2008}.
\end{rmk}

As an ungraded algebra, \( \sotwo(2r) \) is semisimple.
Appealing to the classification of irreducible representations in this case,
we find that the finite-dimensional ungraded irreducible representations are
\[
	V_\lambda\otimes V_\mu 
\]
where \( \lambda \) and \( \mu \) range over the dominant integral weights of \( \so(2r) \)
and for \( V_\lambda \) and \( V_\mu \) the irreducible representations of \( \so(2r) \) with highest weights \( \lambda \) and \( \mu \) respectively (see e.g.\ \cite[Theorem~3.9]{Sepanski2007}).
Note that the action of \( \so(2r)\oplus\so(2r) \) on \( V_\lambda\otimes V_\mu \) is given by
\[
	(x\oplus y)(v\otimes w) = xv\otimes w + v\otimes yw
\]
for \( x,y\in\so(2r), v\in V_\lambda, w\in V_\mu \).

To apply a \( \Ztwo \)-grading to these representations, it is first beneficial to analyse the structure of \( \so(2r) \).

\subsection{Structure of \texorpdfstring{\( \so(2r) \)}{so(2r)}}
Let \( e_{ij} \) be the elementary matrices satisfying \( (e_{ij})_{k\ell} = \delta_{ik}\delta_{j\ell} \)
where \( \delta \) is the Kronecker delta.
A standard basis for \( \so(2r) \) is
\[
	E_{ij} = e_{ij} - e_{ji}, \qquad i,j = 1,\ldots,2r
\]
noting that \( E_{ij} = - E_{ji} \).
(Note that we could interpret \( E_{ij}\in Q\sotwo(2r)Q^{-1} \) as a basis of the first copy \( \so(2r) \). A basis for the second copy can be defined similarly.)
The Lie bracket is given by 
\[
	\comm{E_{ij}}{E_{k\ell}} = 
	\delta_{jk}E_{i\ell} - \delta_{j\ell}E_{ik} - \delta_{ik}E_{j\ell} + \delta_{i\ell}E_{jk}.
\]
We will choose the Cartan subalgebra spanned by the elements \( H_j = -2iE_{2j-1,2j} \) for \( j = 1,\ldots, r \).
The raising/lowering operators for this Cartan basis are
\[
	F^{\zeta\xi}_{jk} = \zeta\xi E_{2j, 2k} - E_{2j-1, 2k-1} - i \xi E_{2j-1, 2k} - i \zeta E_{2j, 2k-1}
\]
where \( j < k \) and \( \zeta,\xi\in\{+,-\} \) are interpreted as \( +1 \) or \( -1 \) in the above expression.
These raising/lowering operators satisfy
\begin{align*}
	\comm{H_j}{F_{jk}^{\zeta\xi}} &= \zeta 2 F_{jk}^{\zeta\xi} &
	\comm{H_k}{F_{jk}^{\zeta\xi}} &= \xi 2 F_{jk}^{\zeta\xi} &
	\comm{H_\ell}{F_{jk}^{\zeta\xi}} &= 0~\text{for}~\ell\neq j,k.
\end{align*}
Let \( \varepsilon_1,\ldots,\varepsilon_r \) be the dual basis of \( H_1,\ldots,H_r \).
The root corresponding to \( F_{jk}^{\zeta\xi} \) is thus \( \alpha_{jk}^{\zeta\xi} \coloneq \zeta 2 \varepsilon_j + \xi 2 \varepsilon_k \).
Let \( \Phi = \{\alpha_{jk}^{\zeta\xi}\mid j<k~\&~\zeta,\xi\in\{+,-\}\} \) be the root system.

For any choice of \( j,k\in\{1,\ldots,r\} \),
\begin{align*}
	&\spn\{H_j,H_k, F_{jk}^{++}, F_{jk}^{+-}, F_{jk}^{-+}, F_{jk}^{--}\}\\
	&\iso \spn\{H_j+H_k,F_{jk}^{++},F_{jk}^{--}\}
	\oplus \spn\{H_j-H_k,F_{jk}^{+-},F_{jk}^{-+}\}
\end{align*}
forms a \( D_2\iso A_1\oplus A_1 \) subalgebra, see \Cref{fig:D2}.

\begin{figure}
	\begin{tikzpicture}
		\begin{scope}[scale=1.414]
			\draw [draw=none, fill=black!10!white] (0,0) -- (0,3.1) -- (3.1,3.1);
			\foreach \i in {-3,...,3} {
				\foreach \j in {-3,...,3} {
					\node [circle, draw, inner sep=0pt, minimum width=3pt] at (\i,\j) {};
				}
			}

			\foreach \i in {0,...,3} {
				\foreach \j in {\i,...,3} {
					\node [circle, fill, inner sep=0pt, minimum width=3pt] at (\i,\j) {};
				}
			}

		\end{scope}

		\begin{scope}[rotate=45]
			\draw [-{>[scale=1.5]}, dashed] (0,0) -- (1,0) node[right]{\( \varepsilon_j \)};
			\draw [-{>[scale=1.5]}, dashed] (0,0) -- (0,1) node[above]{\( \varepsilon_k \)};

			\draw [-{>[scale=1.5]}] (0,0) -- (2,2) node[right]{\( \alpha_{jk}^{++} \)};
			\draw [-{>[scale=1.5]}] (0,0) -- (2,-2) node[right]{\( \alpha_{jk}^{+-} \)};
			\draw [-{>[scale=1.5]}] (0,0) -- (-2,2) node[left]{\( \alpha_{jk}^{-+} \)};
			\draw [-{>[scale=1.5]}] (0,0) -- (-2,-2) node[left]{\( \alpha_{jk}^{--} \)};
		\end{scope}
	\end{tikzpicture}
	\caption{\label{fig:D2}Root system for a \( D_2 \) subalgebra of \( \so(2r) \). The dots are the integral elements, with the black dots being dominant in \( \so(2r) \) (assuming \( k\neq r \)) and the white dots being non-dominant. The shaded area is the cross-section of the fundamental Weyl chamber (assuming \( k\neq r \)). If \( k=r \), then the fundamental Weyl chamber is the entire first quadrant (excluding the axes) and the dominant integral elements are the integral elements inside this quadrant and along the positive half-axes (including zero).}
\end{figure}

Since \( \varepsilon_j = (1/4)(\alpha^{++}_{jk} + \alpha^{+-}_{jk}) \)
and \( \varepsilon_k = (1/4)(\alpha^{++}_{jk} - \alpha^{+-}_{jk}) \),
we find that these two dual vectors are perpendicular and have the same length.
(Indeed, from the \( A_1 \oplus A_1 \) root system, we know that \( \alpha^{++}_{jk} \) and \( \alpha^{+-}_{jk} \) are perpendicular
and, since \( D_r \) is simply laced, \( \alpha^{++}_{jk} \) and \( \alpha^{+-}_{jk} \) are the same length; 
we can use these two facts to show that \( (\varepsilon_j,\varepsilon_k) = 0 \)
and \( (\varepsilon_j,\varepsilon_j) = (\varepsilon_k,\varepsilon_k)  \),
where \( (\,\cdot\,,\,\cdot\,) \) is the inner product induced by the Killing form.)
Therefore, the dot product corresponding to the dual basis \( \varepsilon_1, \ldots, \varepsilon_r \) forms an invariant inner product on the root space,
making the dual basis an orthonormal basis.

The lattice of integral weights is
\[
	P = \{\lambda_1\varepsilon_1 + \cdots + \lambda_r \varepsilon_r\mid \lambda_1,\ldots,\lambda_r\in \ZZ,\, \lambda_1+\lambda_2,\lambda_2+\lambda_3, \ldots, \lambda_{r-1} + \lambda_r\in 2\ZZ\}.
\]
We will impose the lexicographical order on \( P \).
The positive integral weights are then
\[
	P^+ = \{\lambda_1\varepsilon_1 + \cdots + \lambda_r \varepsilon_r \in P \mid \lambda_1 = \cdots = \lambda_{j-1} = 0~\text{and}~\lambda_j >0~\text{for some}~j=1,\ldots,r\}
\]
and the positive roots are
\[
	\Phi^+ = \Phi\cap P^+ = \{2\varepsilon_j \pm 2\varepsilon_k \mid j<k\}.
\]
Thus, the raising operators are \( F_{jk}^{+\xi} \) for \( j<k \) and \( \xi\in\{+,-\} \).
The negative integral weights \( P^- \), negative roots \( \Phi^- \) and lowering operators are defined similarly.
The simple roots are
\[
	\Phi_s = \{2\varepsilon_j - 2\varepsilon_{j+1} \mid j = 1,\ldots,r-1\}
	\cup
	\{2\varepsilon_{r-1} + 2\varepsilon_r\}.
\]
The dominant integral weights are
\[
	\{\lambda_1\varepsilon_1 + \cdots + \lambda_r \varepsilon_r \in P \mid \lambda_1\geq \lambda_2\geq\cdots \geq\lambda_{r-1}\geq\abs{\lambda_r}\}
\]
which is easily verified by taking the dot product with positive roots.

Recall that the dominant integral weights classify the finite-dimensional (ungraded) irreducible representations of \( \so(2r) \). 

\subsection{Twisting \texorpdfstring{\( \so(2r) \)}{so(2r)} representations}
Let \( \lambda = (\lambda_1, \ldots, \lambda_r) \) (in the basis \( \varepsilon_1,\ldots,\varepsilon_r \)) be a dominant integral weight and let \( V_\lambda \) be the corresponding irreducible \( \so(2r) \)-module.
In this subsection,
we will apply a twist in \( \pdual{\Ztwo}/\orthog{(\Ztwo)} = \pdual{\Ztwo} \) and classify the resulting \( \so(2r) \)-module.
We will use this result in the next subsection, together with \Cref{prop:V-=Vifgrade}, to classify the irreducible \( \Ztwo \)-graded \( \sotwo(2r) \)-modules.
Note that the only non-identity element in \( \pdual{\Ztwo} \) is the character \( f\colon\Ztwo\to\units{\CC} \) defined by \( f(\gamma)=(-1)^\gamma \).

Now, consider the \( \ZZ_2 \) grading \( \so(2r) = (\so(2r))_0\oplus(\so(2r))_1 \), where
\begin{align*}
	(\so(2r))_0 &= \spn\{E_{ij} \mid \text{(\( i\leq m \) and \( j\leq m \)) or (\( i>m \) and \( j>m \))}\}\\
	(\so(2r))_1 &= \spn\{E_{ij} \mid \text{(\( i\leq m \) and \( j> m \)) or (\( i>m \) and \( j\leq m \))}\}.
\end{align*}
There are two cases: if \( m \) is even, then the entire Cartan subalgebra is \( 0 \)-graded; if \( m \) is odd, then \( H_{(m+1)/2} \) is \( 1 \)-graded and the other Cartan basis elements are \( 0 \)-graded.

\subsubsection{First case: \texorpdfstring{\( m \)}{m+n} even}
First, recall that the Cartan decomposition yields a \( \ZZ^r \)-grading (called the Cartan grading) where the degree of an element of \( \so(2r) \) is given by its corresponding root in the fundamental weight basis (the degrees of the Cartan elements are all \( 0 \)).
Every irreducible representation is \( \ZZ^r \)-graded, with the degrees corresponding to the weights.

In the case that \( m \) is even, notice that every Cartan element and every raising/lowering operator \( F_{jk}^{\zeta\xi} \) is homogeneous.
Therefore, the Cartan grading refines the \( \Gamma \)-grading.
Hence every irreducible \( \so(2r) \)-module inherits a \( \Gamma \)-grading from the Cartan grading.

More specifically, every weight in the weight lattice \( P \) can be written as a unique linear combination of simple roots.
For each \( \kappa = \kappa_1\varepsilon_1 + \cdots + \kappa_r \varepsilon_r\in P \), observe that
\begin{align*}
	\kappa &= \left(\sum_{j=1}^{r-2}\left(\frac{1}{2}\sum_{k=1}^{j} \kappa_k\right)(2\varepsilon_j - 2\varepsilon_{j+1})\right) \\
	       & \qquad + \frac{1}{4}\left(-\kappa_r + \sum_{k=1}^{r-1}\kappa_k\right)(2\varepsilon_{r-1}-2\varepsilon_{r+1})
	+ \frac{1}{4}\left(\kappa_r + \sum_{k=1}^{r-1}\kappa_k\right)(2\varepsilon_{r-1}+2\varepsilon_{r+1}).
\end{align*}

Set \( K = m/2 \). If \( n>2 \), then the only \( 1 \)-graded simple root is \( \alpha^{+-}_{K,K+1} = 2\varepsilon_K - 2\varepsilon_{K+1} \), while if \( n=2 \) there are two \( 1 \)-graded simple roots: \( \alpha^{+-}_{K,K+1} = 2\varepsilon_{r-1} - 2\varepsilon_{r} \) and \( \alpha^{++}_{K,K+1} = 2\varepsilon_{r-1} + 2\varepsilon_{r} \).
In either case, the sum of coefficients of the \( 1 \)-graded simple roots is
\[
	\frac{1}{2}\sum_{k=1}^K\kappa_k.
\]
Multiplying this sum by two and considering the result modulo four gives the weight lattice \( P \) a \( \ZZ_4 \)-grading.
The roots all have \( \ZZ_4 \)-grading in \( \{0,2\} \); this is exactly the original \( \Ztwo \)-grading.
The weight vectors of an irreducible representation all have \( \ZZ_4 \)-degree in either \( \{0,2\} \) or \( \{1,3\} \) (the orbits of \( \Ztwo\iso \{0,2\} \) acting on \( \ZZ_4 \)), depending on the \( \ZZ_4 \)-degree of the highest weight vector.
A choice of \( \Ztwo \)-set isomorphism \( \{0,2\}\xrightarrow{\iso} \Ztwo \) or \( \{1,3\}\xrightarrow{\iso}\Ztwo \)
yields a \( \Ztwo \)-grading for each irreducible representation.
Note that the two different choices of \( \Ztwo \)-set isomorphism in each case correspond to \( \Ztwo \)-parity shifts of the same grading.

By \Cref{prop:V-=Vifgrade}, we have that \( (V_\lambda)^f\iso V_\lambda \) for every dominant integral weight \( \lambda \).

\subsubsection{Second case: \( m \) odd}
Set \( K = (m+1)/2 \), so that \( H_K \) is \( 1 \)-graded.
Note that the raising/lowering operators with \( K \) in the lower index are no longer homogeneous.
Instead,
\[
	F_{jK}^{\zeta+} + F_{jK}^{\zeta-}
	\qquad \text{and} \qquad
	F_{Kk}^{+\xi} - F_{Kk}^{-\xi} 
\]
are \( 0 \)-graded
while
\[
	F_{jK}^{\zeta+} - F_{jK}^{\zeta-}
	\qquad \text{and} \qquad
	F_{Kk}^{+\xi}+ F_{Kk}^{-\xi}
\]
are \( 1 \)-graded for \( j<K<k \).

\begin{lem}\label{lem:soZ2twist}
	\( (V_\lambda)^f \iso V_{\lambda^f} \) where \( \lambda^f = \lambda - 2 \lambda_r \varepsilon_r \).
\end{lem}
\begin{proof}
	In the notation of~\cite{EK2015},
	we have that \( \mathcal{D} = \FF \)
	and the grading of \( \so(2r) \) is inherited from \( \End_{\mathcal{D}}(W) \)
	where \( W \) is a \( \Ztwo \)-graded \( (m+n) \)-dimensional vector space.
	(This is how we originally defined the grading for \( \so(2r) \).)
	The distinguished element~\cite[Definition~34]{EK2015} in \( \Ztwo \) is
	\[
		r0 + m0 + n1 
	\]
	(choosing the standard basis of \( W \),
	we find that \( g_1 = \cdots = g_m = 0 \) and \( g_{m+1} = \cdots = g_{m+n}= 1 \),
	using the notation of~\cite[Section~7.1]{EK2015}).

	Since we are assuming that \( n \) is odd, the distinguished element is not the identity and thus the grading is outer;
	that is, there exists an element of \( \pdual{\Ztwo} \) whose twist is an outer automorphism of \( \so(2r) \).
	In particular, a twist by \( f \) must be an outer automorphism.

	In determining of the isomorphism class of \( (V_\lambda)^f \),
	an outer automorphism is equivalent to a Dynkin diagram automorphism~\cite[Section~3]{EK2015}.
	Since this grading is a matrix grading,
	this diagram automorphism must be the one
	which swaps the two terminal simple roots \( 2\varepsilon_{r-1} -2\varepsilon_r \)
	and \( 2\varepsilon_{r-1} + 2\varepsilon_r \).
	(For \( r\neq 4 \) this is the only non-trivial diagram automorphism.)

	The dual Cartan basis \( \varepsilon_1,\ldots,\varepsilon_{r-1},\varepsilon_r \) becomes \( \varepsilon_1,\ldots,\varepsilon_{r-1},-\varepsilon_r \)
	after applying this diagram automorphism.
	Thus,
	the highest weight \( \lambda = (\lambda_1,\ldots,\lambda_{r-1},\lambda_r) \) 
	becomes \( (\lambda_1,\ldots,\lambda_{r-1},-\lambda_r) \).
\end{proof}
\begin{rmk}
	We could have also proven the lemma using more direct computations.
	Let \( v \) be the highest weight vector of \( V_\lambda \),
	and let \( \lambda' \) be the weight of \( v \) in \( (V_\lambda)^f \).
	We could then compute \( \lambda^f \) by using the action of the Weyl group
	to map \( \lambda' \) to a dominant integral weight.
\end{rmk}

From the lemma, we deduce that \( (V_\lambda)^f \iso V_\lambda \) if and only if \( \lambda_r = 0 \).


\subsection{First step: \texorpdfstring{\( \Ztwo \)}{Z2}-graded irreducible representations}
We are now ready to classify the finite-dimensional \( \Ztwo \)-graded irreducible representations of \( \sotwo(2r) \).
We will consider the full \( \Ztzt \)-grading in the next subsection.

Let \( \lambda = (\lambda_1,\ldots,\lambda_r) \) and \( \mu = (\mu_1,\ldots,\mu_r) \) be dominant integral weights for \( \so(2r) \).
The irreducible ungraded representations for \( \sotwo(2r) \) are \( V_\lambda\otimes V_\mu \).

If \( m \) is even, then \( V_\lambda\otimes V_\mu \) is \( \Ztwo \)-gradable (induced by the Cartan grading).

Now, assume that \( m \) is odd.
We can easily verify that \( (V_\lambda\otimes V_\mu)^{f_1} = V_\lambda^{f_1} \otimes V_\mu^{f_1} \) in this case.
Therefore, by \Cref{lem:soZ2twist}, we have that \( (V_\lambda\otimes V_\mu)^{f_1} \iso V_\lambda\otimes V_\mu \)
if and only if \( \lambda_r = 0 \) and \( \mu_r = 0 \).
We thus have the following classification:

\begin{lem}\label{lem:so2Z2classification}
	Let \( \lambda = (\lambda_1,\ldots,\lambda_r) \) and \( \mu = (\mu_1,\ldots,\mu_r) \) 
	and set \( \lambda^f = (\lambda_1,\ldots,\lambda_{r-1},-\lambda_r) \) and \( \mu^f = (\mu_1,\ldots,\mu_{r-1},-\mu_r) \).
	Letting \( \lambda \) and \( \mu \) range over the dominant integral weights of \( \so(2r) \),
	the finite-dimensional \( \Ztwo \)-graded irreducible representations of \( \sotwo(2r) \) are as follows.

	For \( m \) even:
	\begin{enumerate}[label=(\arabic*)]
		\item\label{case:Z2meven} The irreducible representations are  \( V_\lambda\otimes V_\mu \) with a \( \Ztwo \)-grading.
			For \( \gamma\in\ZZ_4 \) set
			\[
				W_\gamma = 
				\{v\in V_\lambda\otimes V_\mu \mid v~\text{has weight}~(\kappa,\nu)~\text{with}
				\sum_{k=1}^{m/2} (\kappa_k + \nu_k) = \gamma\pmod{4}\}.
			\]
			Then, up to a \( \Ztwo \)-parity shift, the \( \Ztwo \)-grading is either \( W_0\oplus W_2 \) or \( W_1\oplus W_3 \) depending on the \( \ZZ_4 \)-degree of the highest weight.
	\end{enumerate}

	For \( m \) odd:
	\begin{enumerate}[resume, label=(\arabic*)]
		\item\label{case:Z2moddloop} If \( \lambda_r \neq 0 \) or \( \mu_r \neq 0 \): the irreducible representations are
			\[
				\loopmodule[\Ztzt/\{00,01\}]{V_\lambda\otimes V_\mu}{\Ztzt/\{00,01\}} 
				= \loopmodule[\Ztzt/\{00,01\}]{V_{\lambda^f}\otimes V_{\mu^f}}{\Ztzt/\{00,01\}}.
			\]
		\item\label{case:Z2moddgrading} If \( \lambda_r=\mu_r = 0 \): the irreducible representations are \( V_\lambda\otimes V_\mu \) with a \( \Ztwo \)-grading.
			Let \( \psi_\lambda\colon V_\lambda \to (V_\lambda)^{f_1} \) and \( \psi_\mu\colon V_\mu \to (V_\mu)^{f_1}  \) 
			be the isomorphisms such that \( (\psi_\lambda)^2 = (\psi_\mu)^2 = \Id \)  (c.f.\ \Cref{lem:soZ2twist}).
			Then the grading is
			\begin{align*}
				(V_\lambda\otimes V_\mu)_0 &= (\Id+\psi_\lambda)V_\lambda\otimes (\Id+\psi_\mu)V_\mu + (\Id-\psi_\lambda)V_\lambda\otimes (\Id-\psi_\mu)V_\mu \\
				(V_\lambda\otimes V_\mu)_1 &= (\Id+\psi_\lambda)V_\lambda\otimes (\Id-\psi_\mu)V_\mu + (\Id-\psi_\lambda)V_\lambda\otimes (\Id+\psi_\mu)V_\mu 
			\end{align*}
			or a \( \Ztwo \)-parity shift of this grading.
	\end{enumerate}
\end{lem}

This lemma is an immediate application of \Cref{thm:bijection} and \Cref{prop:V-=Vifgrade}.
The grading in the third case follows from \Cref{rmk:V-=Vifgradeproj},
noting that \( \psi_\lambda\otimes\psi_\mu\colon V_\lambda\otimes V_\mu \to (V_\lambda\otimes V_\mu)^{f_1} \) is an isomorphism and that
\begin{align*}
	2(\Id\otimes\Id + \psi_\lambda\otimes\psi_\mu) &= (\Id+\psi_\lambda)\otimes (\Id+\psi_\mu) + (\Id-\psi_\lambda)\otimes (\Id-\psi_\mu) \\
	2(\Id\otimes\Id - \psi_\lambda\otimes\psi_\mu) &= (\Id+\psi_\lambda)\otimes (\Id-\psi_\mu) + (\Id-\psi_\lambda)\otimes (\Id+\psi_\mu).
\end{align*}

\subsection{Second step: \texorpdfstring{\( \Ztzt \)}{Z2xZ2}-graded irreducible representations and recolouring}
In this subsection we iterate \Cref{thm:bijection} to classify all of the \( \Ztzt \)-graded irreducible representations of \( \sotwo(2r) \) and the colour algebra \( \sotwoc(m,n) \).

Let \( H_j = -2i Q^{-1}E_{2j-1,2j}Q  \) for \( j=1,\ldots,2r \).
Then \( H_j \) for \( j=1,\ldots,r \) forms a Cartan basis for the first copy of \( \so(2r) \)
in \( \sotwo(2r) \)
and \( H_j \) for \( j=r+1,\ldots,2r \) forms a Cartan basis for the second copy of \( \so(2r) \).
With respect to the \( \Ztzt \)-grading, none of these basis elements are homogeneous.
However, the elements \( H_j \pm H_{r+j} \) for \( j=1,\ldots,r \) are homogeneous.
If \( m \) is odd, then \( H_K+H_{r+K} \) has degree \( 10 \) and \( H_K - H_{r+K} \) has degree \( 11 \)
for \( K = (m+1)/2 \).
Otherwise, \( H_j + H_{r+j} \) has degree \( 00 \) and \( H_j - H_{r+j} \) has degree \( 01 \).

Choose \( f_2\in\pdual{\Ztzt} \) such that \( f_2(\gamma_1\gamma_2) = (-1)^{\gamma_2} \).
Then \( \pdual{\Ztzt}/\orthog{H} = \{\id \orthog{H},f_2\orthog{H}\} \)
for
\( H = \{00,01\} \).

Let \( \rho \) be the representation corresponding to \( V_\lambda\otimes V_\mu \).
Let \( v\in V_\lambda\otimes V_\mu \) be any weight vector, and let its weight be \( (\kappa,\nu) = (\kappa_1,\ldots,\kappa_r,\nu_1,\ldots,\nu_r) \) (so that \( \kappa \) and \( \nu \) are \( \so(2r) \)-weights).
Then, for any \( j = 1,\ldots,r \),
\begin{align*}
	\rho^{f_2}(H_j)v
	&= \left(\frac{1}{2}\rho^{f_2}(H_j + H_{r+j}) + \frac{1}{2}\rho^{f_2}(H_j - H_{r+j})\right)v\\
	&= \left(\frac{1}{2}\rho(H_j + H_{r+j}) - \frac{1}{2}\rho(H_j - H_{r+j})\right)v\\
	&= \rho(H_{r+j})v\\
	&= \nu_j v
\end{align*}
Similarly, \( \rho^{f_2}(H_{r+j}) v = \kappa_j v \).
Therefore, twisting by \( f_2 \) maps weights \( (\kappa,\nu) \mapsto (\nu,\kappa) \).
Consequently, \( (\mu,\lambda) \) is the highest weight of the twisted module, so
\[
	(V_\lambda\otimes V_\mu)^{f_2} \iso V_\mu\otimes V_\lambda
\]
as ungraded modules.

If \( m \) is odd and
\( \lambda_r\neq 0 \) or \( \mu_r \neq 0 \),
then
\[
	(\loopmodule[\Ztzt]{V_\lambda\otimes V_\mu}{\{00,01\}})^{f_2} = \loopmodule[\Ztzt]{(V_\lambda\otimes V_\mu)^{f_2}}{\{00,01\}} \iso \loopmodule[\Ztzt]{V_\mu\otimes V_\lambda}{\{00,01\}}
\]
is isomorphic to \( \loopmodule[\Ztzt]{V_\lambda\otimes V_\mu}{\{00,01\}} \)
if and only if \( \mu \in\{\lambda, \lambda^f\} \) (with \( \lambda^f \) as in \Cref{lem:so2Z2classification}).

If \( m \) is even or if \( m \) is odd but \( \lambda_r = \mu_r = 0 \), then
\[
	(V_\lambda\otimes V_\mu)^{f_2} \iso V_\mu\otimes V_\lambda
\]
as \( \Ztwo \)-graded modules if and only if \( \mu = \lambda \).
This follows by noting that the gradings for \( V_\lambda\otimes V_\lambda \) in \Cref{lem:so2Z2classification} remains unchanged after twisting by \( f_2 \).

By applying \Cref{thm:bijection}, we obtain a classification of the \( \Ztzt \)-graded irreducible \( \sotwo(2r) \)-modules.

\begin{thm}\label{thm:so2Ztztclassification}
	Let \( \lambda = (\lambda_1,\ldots,\lambda_r) \) and \( \mu = (\mu_1,\ldots,\mu_r) \) 
	and set \( \lambda^f = (\lambda_1,\ldots,\lambda_{r-1},-\lambda_r) \) and \( \mu^f = (\mu_1,\ldots,\mu_{r-1},-\mu_r) \).
	Let \( \Psi\colon (V_\lambda\otimes V_\mu)\to (V_\mu\otimes V_\lambda)^{f_2} \) be an isomorphism
	such that \( \Psi^2 = \Id \).
	Letting \( \lambda \) and \( \mu \) range over the dominant integral weights of \( \so(2r) \),
	the finite-dimensional \( \Ztzt \)-graded irreducible representations of \( \sotwo(2r) \) are as follows.

	For \( m \) even:
	\begin{enumerate}[label=(\arabic*)]
		\item\label{case:Ztztmevenloop} If \( \mu\neq \lambda \): the irreducible representations are
			\[
				\loopmodule[\Ztzt]{V_\lambda\otimes V_\mu}{\{00,01\}} \iso \loopmodule[\Ztzt]{V_\mu\otimes V_\lambda}{\{00,01\}}
			\]
		where \( V_\lambda\otimes V_\mu \) has a \( \Ztwo\iso\Ztzt/\{00,01\} \) grading as in \Cref{lem:so2Z2classification} case~\ref{case:Z2meven}.

		\item\label{case:Ztztmevengrade} If \( \mu = \lambda \): the irreducible representations are \( V_\lambda \otimes V_\mu \) with a \( \Ztzt \)-grade.
			If \( V_\lambda\otimes V_\mu = W_0\oplus W_1 \) is the \( \Ztwo \)-grading from \Cref{lem:so2Z2classification} case~\ref{case:Z2meven}, then the \( \Ztzt \)-grading is
			\begin{align*}
				 (V_\lambda \otimes V_\mu)_{00}
				 &= (\Id+\Psi) W_0
				  & (V_\lambda \otimes V_\mu)_{11}
				  &= (\Id-\Psi)W_0\\
				 (V_\lambda \otimes V_\mu)_{10}
				 &= (\Id+\Psi) W_1
				  & (V_\lambda \otimes V_\mu)_{01}
				  &= (\Id-\Psi) W_1
			\end{align*}
			or any \( \Ztzt \)-parity shifts.
	\end{enumerate}
	For \( m \) odd:
	\begin{enumerate}[resume, label=(\arabic*)]
		\item\label{case:Ztztmoddlooploop} If \( \lambda_r \neq 0 \) or \( \mu_r\neq 0 \) and if \( \mu\notin\{\lambda,\lambda^f\} \): the irreducible representations are
			\[
				\loopmodule[\Ztzt]{V_\lambda\otimes V_\mu}{\Ztzt}
				\iso\loopmodule[\Ztzt]{V_{\lambda^f}\otimes V_{\mu^f}}{\Ztzt}
				\iso \loopmodule[\Ztzt]{V_\mu\otimes V_\lambda}{\Ztzt}
				\iso\loopmodule[\Ztzt]{V_{\mu^f}\otimes V_{\lambda^f}}{\Ztzt}.
			\]
		\item\label{case:Ztztmoddloopgrade} If \( \lambda_r \neq 0 \) or \( \mu_r \neq 0 \) and if \( \mu\in\{\lambda,\lambda^f\} \): the irreducible representations are
			\[
				\loopmodule[\Ztzt/\{00,01\}]{V_\lambda\otimes V_\mu}{\Ztzt/\{00,01\}}
				\iso \loopmodule[\Ztzt/\{00,01\}]{V_{\lambda^f}\otimes V_{\mu^f}}{\Ztzt/\{00,01\}}
			\]
			with a \( \Ztzt \)-grading.
			If \( \mu = \lambda \), then the \( \Ztzt \)-grading is
			\begin{align*}
				(V_\lambda \otimes V_\mu)_{00}
				 &= (\Id+\Psi) (V_\lambda\otimes V_\mu) \otimes e_0 
				 & (V_\lambda \otimes V_\mu)_{11}
				 &= (\Id-\Psi) (V_\lambda\otimes V_\mu)\otimes e_0\\
				 (V_\lambda \otimes V_\mu)_{10}
				 &= (\Id+\Psi) (V_\lambda\otimes V_\mu)\otimes e_1
				 & (V_\lambda \otimes V_\mu)_{01}
				 &= (\Id-\Psi) (V_\lambda\otimes V_\mu)\otimes e_1
			\end{align*}
			or a \( \{00,01\} \)-parity shift of this grading.
			If \( \mu= \lambda^f \),
			then the grading is the same as above except that the \( 01 \)- and \( 10 \)-sectors are swapped (note that we identify \( V_\lambda \) and \( V_{\lambda^f} \iso (V_\lambda)^{f_1} \) as vector spaces).
		\item\label{case:Ztztmoddgradeloop} If \( \lambda_r = \mu_r = 0 \) and if \( \mu \neq \lambda \): the irreducible representations are
			\[
				\loopmodule[\Ztzt]{V_\lambda\otimes V_\mu}{\{00,01\}}
				\iso \loopmodule[\Ztzt]{V_\mu\otimes V_\lambda}{\{00,01\}}
			\]
			where \( V_\lambda\otimes V_\mu \) has a \( \Ztwo\iso\Ztzt/\{00,01\} \)-grading as in \Cref{lem:so2Z2classification} case~\ref{case:Z2moddgrading}.
		\item\label{case:Ztztmoddgradegrade} If \( \lambda_r = \mu_r = 0 \) and if \( \mu = \lambda \):
			the irreducible representations are \( V_\lambda\otimes V_\lambda \) with a \( \Ztzt \)-grading.
			If \( V_\lambda\otimes V_\mu = W_0\oplus W_1 \) is the \( \Ztwo \)-grading from \Cref{lem:so2Z2classification} case~\ref{case:Z2moddgrading}, then the \( \Ztzt \)-grading is
			\begin{align*}
				 (V_\lambda \otimes V_\mu)_{00}
				 &= (\Id+\Psi) W_0
				  & (V_\lambda \otimes V_\mu)_{11}
				  &= (\Id-\Psi) W_0\\
				 (V_\lambda \otimes V_\mu)_{10}
				 &= (\Id+\Psi) W_1
				  & (V_\lambda \otimes V_\mu)_{01}
				  &= (\Id-\Psi) W_1
			\end{align*}
			or any \( \Ztzt \)-parity shifts of this grading.
	\end{enumerate}
\end{thm}
The theorem follows immediately from the discussion preceding it.
For the third case, note that \( \loopmodule[\Ztzt]{\loopmodule[\Ztzt/\{00,01\}]{\cdot}{\Ztzt/\{00,01\}}}{\{00,01\}} \iso \loopmodule[\Ztzt]{\cdot}{\Ztzt} \).
For the fourth case, if \( \mu = \lambda^f \) then \( v\otimes w\otimes e_\gamma \mapsto f(\gamma) \Psi(v\otimes w)\otimes e_\gamma \) is the isomorphism \( \loopmodule[\Ztzt/\{00,01\}]{V_\lambda\otimes V_{\lambda^f}}{\Ztzt/\{00,01\}} \to \loopmodule[\Ztzt/\{00,01\}]{V_{\lambda^f}\otimes V_\lambda}{\Ztzt/\{00,01\}} \) (for use in the projection formula in \Cref{rmk:V-=Vifgradeproj}); the \( f(\gamma) \) factor causes the \( 01 \)- and \( 10 \)-sectors to be swapped.

By recolouring with the multiplier \( \sigma(\alpha_1\alpha_2, \beta_1\beta_2) = (-1)^{\alpha_2\beta_1} \), \Cref{thm:so2Ztztclassification} immediately yields a classification of the \( \Ztzt \)-graded irreducible representations for the Lie colour algebra \( \sotwoc(m,n) \).

\subsection{Third step: \texorpdfstring{\( \Ztwo \)}{Z2}-graded irreducible representations for the colour algebra}
With a classification of the graded representations for the colour algebra \( \sotwoc(m,n) \),
we can iterate the inverse of the bijection in \Cref{thm:bijection} to also classify the ungraded representations.

First consider coarsening to a \( \Ztzt/\{00,01\} \)-grading.
To use \Cref{prop:V+h=Vifred}, we need to consider parity shifts by \( 01 \).

First observe that, for a \( \Ztzt \)-graded \( \sotwo(2r) \)-module \( V \) (with representation \( \rho \)),
we have \( V^\sigma \iso (V^\sigma)^{+h} \) if and only if \( V\iso ((V^\sigma)^{+h})^\sigma \).
(Indeed \( V^\sigma \iso (V^\sigma)^{+h} \implies  V= V^{\sigma\sigma}\iso ((V^\sigma)^{+h})^\sigma \implies V^\sigma\iso ((V^\sigma)^{+h})^{\sigma\sigma} = (V^\sigma)^{+h}  \).)
Take \( x\in\g_\alpha \) and \( v\in V_\xi \).
Then, for our choice of multiplier
\( \sigma(\alpha_1\alpha_2, \beta_1\beta_2) = (-1)^{\alpha_2\beta_1} \)
we have
\[
	((\rho^\sigma)^{+01})^\sigma(x) v
	=\sigma(\alpha,\xi)(\rho^{+01})^\sigma(x) v
	=\sigma(\alpha,\xi)\sigma(\alpha, \xi+01)\rho(x) v
	= \rho(x) v.
\]
Thus, in this case, \( ((V^\sigma)^{+01})^\sigma = V^{+01} \).
We deduce that \( V \iso V^{+01}  \)
if and only if \( V\iso ((V^\sigma)^{+01})^\sigma \)
if and only if \( V^\sigma \iso (V^\sigma)^{+01} \).
This implies,
by \Cref{prop:V+h=Vifred},
that \( V \) is \( \Ztwo \)-graded reducible
if and only if \( V^\sigma \) is \( \Ztwo \)-graded reducible
(assuming \( V \) is \( \Ztzt \)-graded irreducible).
Moreover,
since the \( H \)-parity shift isomorphisms \( V\to V^{+01} = ((V^\sigma)^{+01})^\sigma \) and \( V^\sigma\to (V^\sigma)^{+01} \) (if they exist) are identical as linear maps,
corresponding irreducible \( \Ztwo \)-graded \( \sotwo(2r) \)- and \( \sotwoc(2r) \)-modules
defined via \Cref{prop:V+h=Vifred}
are equal as vector subspaces of \( V = V^\sigma \).
Thus,
to classify the \( \Ztwo \)-graded \( \sotwoc(2r) \)-modules,
we can
simply twist the irreducible \( \Ztwo \)-graded \( \sotwo(2r) \)-modules
in \Cref{lem:so2Z2classification}
.

Notice that in coarsening the grading, we are mapping \( \Ztzt\to\Ztwo \) by \( \beta_1\beta_2\mapsto \beta_1 \).
Therefore, the multiplier \( \sigma \) 
factors through \( (\Ztzt)\times \Ztwo \)
to define a map \( \tau\colon (\Ztzt)\times\Ztwo\to\units{\FF} \) by
\[
	\tau (\alpha_1\alpha_2, \beta) = (-1)^{\alpha_2 \beta}.
\]
It is easy to verify that after applying \Cref{prop:V+h=Vifred} the representation becomes
\( \rho^{\tau} \) defined by
\[
	\rho^{\tau}(x)v
	=\tau(\alpha,\Lambda)\rho^{}(x)v
\]
for \( x\in\g_\alpha \) and \( v\in V_\Lambda \) (\( \alpha\in\Ztzt, \, \Lambda\in\Ztwo \)).
 
Thus, we have the following classification of \( \Ztwo \)-graded irreducible \( \sotwoc(m,n) \) modules.
\begin{lem}\label{lem:so2cZ2classification}
	The \( \Ztwo \)-graded irreducible representations of the colour algebra \( \sotwoc(m,n) \) are the same as those in \Cref{lem:so2Z2classification},
	recoloured with 
	the function \( \tau\colon(\Ztzt)\times\Ztwo\to\units{\FF} \)
	defined by \( \tau(\alpha_1\alpha_2, \beta) = (-1)^{\alpha_2\beta} \)
	(\( \tau \) is induced from \( \sigma \) by passing to the quotient \( \Ztzt/\{00,01\} \) in the second variable).
\end{lem}

\subsection{Fourth step: ungraded irreducible representations for the colour algebra}
We are now ready to classify the ungraded irreducible representations for the colour algebra \( \sotwoc(m,n) \).
In this subsection, we coarsen the grading from \( \Ztwo\iso \Ztzt/\{00,01\} \) to the trivial group.

Like in the previous section, \( V^\tau \iso (V^\tau)^{+h} \) if and only if \( V \iso ((V^\tau)^{+h})^\tau \).
Let \( x\in\g_\alpha \) and \( v\in V_\Lambda \) (for \( \alpha\in\Ztzt \) and \( \Lambda\in\Ztwo \)).
Choosing \( h=1 \) and \( \tau\colon(\Ztzt)\times\Ztwo\to\FF \) with \( \tau(\alpha_1\alpha_2, \beta) = (-1)^{\alpha_2\beta} \), we find
\begin{align*}
	((\rho^\tau)^{+1})^\tau(x)v
	&=\tau(\alpha,\Lambda)((\rho)^{+1})^\tau(x)v\\
	&=\tau(\alpha,\Lambda)((\rho)^{+1})^\tau(x)v\\
	&= (-1)^{\alpha_2\Lambda} (-1)^{\alpha_2(\Lambda +1)} \rho(x)v\\
	&= (-1)^{\alpha_2} \rho(x)v\\
	&= f_2(\alpha)\rho(x)v.
\end{align*}
Therefore, in this case, \( ((V^\tau)^{+1})^\tau \iso (V^{f_2})^{+1}  \).

\begin{thm}
	Let \( \lambda = (\lambda_1,\ldots,\lambda_r) \) and \( \mu = (\mu_1,\ldots,\mu_r) \) 
	and set \( \lambda^f = (\lambda_1,\ldots,\lambda_{r-1},-\lambda_r) \) and \( \mu^f = (\mu_1,\ldots,\mu_{r-1},-\mu_r) \).
	Let \( \Psi\colon (V_\lambda\otimes V_\mu)\to (V_\mu\otimes V_\lambda)^{f_2} \) be an isomorphism
	such that \( \Psi^2 = \Id \).
	Let \( \tau\colon(\Ztzt)\times\Ztwo\to\units{\FF} \)
	be the function \( \tau(\alpha_1\alpha_2, \beta) = (-1)^{\alpha_2\beta} \) for recolouring.

	Letting \( \lambda \) and \( \mu \) range over the dominant integral weights of \( \so(2r) \),
	the finite-dimensional ungraded irreducible representations of \( \sotwoc(m,n) \) are as follows.

	For \( m \) even:
	\begin{enumerate}[label=(\arabic*)]
		\item the irreducible representations are \( (V_\lambda\otimes V_\mu)^\tau \).
			The recolouration by \( \tau \) is defined with respect to the \( \Ztwo \)-grading in \Cref{lem:so2Z2classification} case~\ref{case:Z2meven}.
	\end{enumerate}

	For \( m \) odd:
	\begin{enumerate}[resume, label=(\arabic*)]
		\item If \( \lambda_r \neq 0 \) or \( \mu_r \neq 0 \) and \( \mu\notin \{\lambda,\lambda^f\} \): the irreducible representations are 
			\[
				(\loopmodule[\Ztzt/\{00,01\}]{V_\lambda\otimes V_\mu}{\Ztzt/\{00,01\}})^\tau
				\iso (\loopmodule[\Ztzt/\{00,01\}]{V_{\lambda^f}\otimes V_{\mu^f}}{\Ztzt/\{00,01\}})^\tau.
			\]
			The recolouration by \( \tau \) is defined with respect to the \( \Ztwo \)-grading of the loop module.
		\item\label{case:ungradedmoddunloop} If \( \lambda_r \neq 0 \) or \( \mu_r \neq 0 \) and \( \mu\in \{\lambda,\lambda^f\} \): the irreducible representations are
			\begin{align*}
				\{v\otimes w\otimes e_0 + \Psi(v\otimes w) \otimes e_1 \mid v\otimes w\in V_\lambda\otimes V_\mu\}
				&\subseteq
				(\loopmodule[\Ztzt/\{00,01\}]{V_\lambda\otimes V_\mu}{\Ztzt/\{00,01\}})^\tau\\
				\{v\otimes w\otimes e_0 - \Psi(v\otimes w) \otimes e_1 \mid v\otimes w\in V_\lambda\otimes V_\mu\}
				&\subseteq
				(\loopmodule[\Ztzt/\{00,01\}]{V_\lambda\otimes V_\mu}{\Ztzt/\{00,01\}})^\tau.
			\end{align*}
		\item If \( \lambda_r = \mu_r = 0 \): the irreducible representations are \( (V_\lambda\otimes V_\mu)^\tau \).
			The recolouration by \( \tau \) is defined with respect to the \( \Ztwo \)-grading in \Cref{lem:so2Z2classification} case~\ref{case:Z2moddgrading}.
	\end{enumerate}
\end{thm}
\begin{proof}
	Consider the case where \( m \) is even (case~\ref{case:Z2meven} in \Cref{lem:so2Z2classification}).
	In this case, the grading is induced from the Cartan grading.
	For the Cartan grading of a Lie algebra, irreducible is equivalent to graded irreducible.
	The result then follows because
	the action of the Cartan subalgebra (which defines the grading) does not change under recolouration.

	Now assume that \( m \) is odd
	and consider the case where \( \lambda_r\neq 0 \) or \( \mu_r \neq 0 \) (case~\ref{case:Z2moddloop} in \Cref{lem:so2Z2classification}).
	If \( \mu\notin \{\lambda, \lambda^f\} \) then
	\begin{align*}
		(((\loopmodule[\Ztzt/\{00,01\}]{V_\lambda\otimes V_\mu}{\Ztzt/\{00,01\}})^\tau)^{+1})^\tau
		&\iso (\loopmodule[\Ztzt/\{00,01\}]{(V_\lambda\otimes V_\mu)^{f_2}}{\Ztzt/\{00,01\}})^{+1}\\
		&\iso (\loopmodule[\Ztzt/\{00,01\}]{V_\mu \otimes V_\lambda}{\Ztzt/\{00,01\}})^{+1}\\
		&\iso \loopmodule[\Ztzt/\{00,01\}]{V_\mu \otimes V_\lambda}{\Ztzt/\{00,01\}}
	\end{align*}
	cannot be isomorphic to \( \loopmodule[\Ztzt/\{00,01\}]{V_\lambda\otimes V_\mu}{\Ztzt/\{00,01\}} \).
	That is, the \( \Ztwo \)-parity shifts are not isomorphic in this case,
	so by \Cref{prop:V+h=Vifred} the recoloured loop module is irreducible.

	Still considering the case \( \lambda_r\neq 0 \) or \( \mu_r \neq 0 \),
	suppose now that \( \mu\in\{\lambda,\lambda^f\} \).
	Then the map \( v\otimes w \otimes e_\gamma \mapsto \Psi(v\otimes w) \otimes e_\gamma \)
	(for \( v\otimes w\in V_\lambda\otimes V_\mu \) and \( \gamma \in \Ztzt/\{00,01\} \))
	is an isomorphism
	\[
		\loopmodule[\Ztzt/\{00,01\}]{V_\lambda\otimes V_\mu}{\Ztzt/\{00,01\}}
		\to (\loopmodule[\Ztzt/\{00,01\}]{(V_\lambda\otimes V_\mu)^{f_2}}{\Ztzt/\{00,01\}})^{+1}.
	\]
	Applying \Cref{prop:V+h=Vifred}, \Cref{rmk:V+h=Vifredproj} and recolouring,
	we find the irreducible representations of the form in case~\ref{case:ungradedmoddunloop} of this theorem.

	Finally, consider the case where \( \lambda_r = \mu_r = 0 \)
	(case~\ref{case:Z2moddgrading} in \Cref{lem:so2Z2classification}).
	If \( \mu\neq\lambda  \) then \( (((V_\lambda\otimes V_\mu)^\tau)^{+1})^\tau\iso ((V_\lambda\otimes V_\mu)^{f_2})^{+1}\iso (V_\mu \otimes V_\lambda)^{+1} \) cannot be (\( \Ztwo \)-graded) isomorphic to \( V_\lambda\otimes V_\mu \).
	So assume \( \mu = \lambda \).
	Let \( v\otimes v \) be the highest weight vector of \( V_\lambda\otimes V_\lambda \).
	Notice that the two possible gradings (the two \( \Ztwo \)-parity shifts) in 
	\Cref{lem:so2Z2classification} case~\ref{case:Z2moddgrading}
	can be distinguished by the degree of \( (\Id\otimes\Id + \psi_\lambda\otimes \psi_\lambda)(v\otimes v) \).
	But the highest weight \( (\lambda,\lambda) \) remains unchanged after twisting by \( f_2 \).
	Therefore, \( v\otimes v \) is still the highest weight vector of \( (V_\lambda\otimes V_\lambda)^{f_2} \); by applying \( \Id\otimes\Id + \psi_\lambda\otimes\psi_\lambda \) to this highest weight vector, we find that the grading remains unchanged.
	That is,
	\( (V_\lambda\otimes V_\lambda)^{f_2} \iso V_\lambda\otimes V_\lambda \)
	as \( \Ztwo \)-graded modules and hence 
	\( ((V_\lambda\otimes V_\lambda)^{f_2})^{+1} \niso V_\lambda\otimes V_\lambda \)
	(noting that the two \( \Ztwo \)-parity shifts of \( V_\lambda\otimes V_\lambda \) are not isomorphic by \Cref{prop:V+h=Vifred}).
	We deduce that \( ((V_\lambda\otimes V_\lambda)^\tau)^{+1}\niso (V_\lambda\otimes V_\lambda)^\tau \)
	and hence \( (V_\lambda \otimes V_\lambda)^\tau \) is irreducible as an ungraded module by \Cref{prop:V+h=Vifred}.
\end{proof}

\section{Conclusions}\label{sec:conclusion}
The loop module is a natural structure to consider in the study of Lie colour algebras,
including their applications to quantum mechanical systems.
In this direction, we showed that the Hilbert space of the \( \Ztzt \)-graded quantum mechanical system of~\cite{BD2020a} can be obtained via a loop module.
The loop module could potentially be a useful tool for constructing quantum mechanical systems in future work.

The main result of this paper (\Cref{thm:bijection}) is the existence of a bijection between finite-dimensional \( G/M \)- and \( G/N \)-graded irreducible representations for \( N \leq M \leq G \) with \( M/N \) cyclic of prime order (up to equivalence relations weaker than isomorphism).
This bijection can be iterated following the cyclic groups of prime order in the Jordan--H\"older decomposition
to yield a bijection between ungraded and \( G \)-graded irreducible representations
for any finite abelian group \( G \).
A similar bijection had previously been noted in~\cite[Remark~7.2]{EK2017},
but the iterated loop module by cyclic groups of prime order simplifies the construction.
In combination with discolouration/recolouration, this bijection provides a general procedure for generating all of the \( G \)-graded Lie colour algebra representations from the well-studied \( \Ztwo \)-graded Lie superalgebra representations.


As an example of applying the bijection, we constructed all of the irreducible representations for the Lie colour algebra \( \splc[2] \).
These irreducible representations agree with those previously derived in \cite{CSVO2006,BS2022} using a different construction.
Despite the bijection of \Cref{thm:bijection},
these \( \Ztzt \)-graded Lie colour algebra representations have a remarkably different structure from those of the ungraded discoloured Lie algebra.
As a more involved example, we also classified the irreducible representations for the Lie colour algebra \( \sotwoc(m,n) \).

The bijection (\Cref{thm:bijection}) also has applications to the study of Lie colour algebra quantum mechanical systems.
For example, we could potentially apply this procedure to construct the irreducible representations of \( \Ztzt \)-supersymmetry algebras with \( \mathcal{N}\geq 2 \), whose classification is desired for the quantum mechanical systems studied in~\cite{AD2022}.
Future research should work to obtain a general classification of irreducible Lie colour algebra representations, possibly by applying the bijection of this paper to existing classifications.


\subsection*{Acknowledgements}
This work was partially supported by the Australian Research Council through Discovery Project DP200101339.
The author would like to thank Kenneth Price for making him aware of~\cite{Price2024},
and an anonymous reader for directing him to the loop module construction of~\cite{MZ2018}.
The reviewer for this article gave extremely valuable feedback
and their suggestions greatly improved this manuscript.
The author is especially grateful to Phillip Isaac for his invaluable guidance, and for taking the time to proofread to the manuscript.


\bibliographystyle{elsarticle-num}

\end{document}